\documentclass[arxiv, preprint]{imsart}

\RequirePackage[OT1]{fontenc}
\RequirePackage{amsthm,amsmath}
\RequirePackage{natbib}
\RequirePackage[colorlinks,citecolor=blue,urlcolor=blue]{hyperref}
\RequirePackage{amssymb, enumerate, graphicx, verbatim, multirow, color}

\startlocaldefs
\numberwithin{equation}{section}
\theoremstyle{plain}
\newtheorem{theorem}{Theorem}[section]

\newtheorem{lemma}[theorem]{Lemma}
\newtheorem{remark}[theorem]{Remark}
\newtheorem{assumption}[theorem]{Assumption}
\newtheorem{condition}{Condition}
\endlocaldefs

\newcommand{\by}{\mathbf{y}}
\newcommand{\bz}{\mathbf{z}}

\newcommand{\bE}{\textsf{E}}
\newcommand{\bP}{\textsf{P}}

\newcommand{\bV}{\textsf{Var}}
\newcommand{\bCov}{\textsf{Cov}}
\newcommand{\bvareps}{\boldsymbol\varepsilon}

\newcommand{\CBP}{\textrm{CBP}}
\newcommand{\CBB}{\textrm{CBB}}
\newcommand{\per}{\textrm{P}}

\def\ep{\textsf{E}} 
\def\Var{\textsf{Var}} 

\graphicspath{ {image/} {../code/2019/} }

\begin{document}

\begin{frontmatter}
\title{Change-point detection for multivariate and non-Euclidean data with local dependency}
\runtitle{Change-point detection for data with local dependency}

\begin{aug}
\author{\fnms{Hao} \snm{Chen} \ead[label=e1]{hxchen@ucdavis.edu}}
%
\runauthor{Hao Chen}

\affiliation{University of California, Davis}

\address{Department of Statistics\\
University of Calfornia, Davis\\
One Shields Avenue \\
Davis, Calfornia 95616 \\
USA \\
\printead{e1}\\
\phantom{E-mail:\ }}

%
%
\end{aug}

\begin{abstract}
In a sequence of multivariate observations or non-Euclidean data objects, such as networks, local dependence is common and could lead to false change-point discoveries.  We propose a new way of permutation -- circular block permutation with a random starting point -- to address this problem.  This permutation scheme is studied on a non-parametric change-point detection framework based on a similarity graph constructed  on the observations, leading to a general framework for change-point detection for data with local dependency.  Simulation studies show that this new framework retains the same level of power when there is no local dependency, while it controls type I error correctly for sequences with and without local dependency.  We also derive an analytic $p$-value approximation under this new framework.  The approximation works well for sequences with length in hundreds and above, making this approach fast-applicable for long data sequences.
\end{abstract}



\end{frontmatter}

\section{Introduction}
 
Change-point detection is a widely studied problem in statistics and has its applications in many fields. 
In the typical formulation, we have observations $\{y_t: t=1,\dots,n\}$ over time (or some other meaningful orderings, such as a one-dimensinoal spatial domain), and test whether there exists $\tau\in \{1,\dots,n-1\}$ such that the underlying distribution of $y_t$ changes at $\tau$. There is a rich literature of this model when $y_t$'s are real or integer valued scalars (see \cite{carlstein1994change,csorgo1997limit} for a survey).

As we entering the big data era, change-point analysis for sequences of multivariate observations or non-Euclidean data objects is gaining more and more attentions.  For example, in text or sequence analysis, each observation in the sequence could be a vector of word counts over a large dictionary of words \citep{giron2005bayesian, tsirigos2005new}.  Network data is ubiquitous nowadays as well.  Email, phone and online chat records can be used to construct networks of social interactions among individuals \citep{kossinets2006empirical,eagle2009inferring}.   
A large part of these studies is characterizing how the network evolves through time.  Here, each observation is a network and one might ask whether there is an abrupt shift in network connectivity at any point in  time.  In these sequences of complicated data types, it is common that observations are autocorrelated.  For example, relationships among people last over an extended time period and the social networks have serial correlations.


A closely related field is time series data analysis.  The ARCH model proposed by \cite{engle1982autoregressive} and the GARCH model proposed by \cite{bollerslev1986generalized} and their variants were widely used for studying one-dimensional time series data.  There are many generalizations to accommodate multivariate time series data (see for examples \cite{bauwens2006multivariate,silvennoinen2009multivariate,aue2009break} and references therein).  These models are useful for low-dimensional data and/or for detecting specific types of changes.  For high-dimensional data, tests based on these parametric models cannot be applied or lack  power unless some strong assumptions are made on the data 
to avoid the estimation of the large number of nuisance parameters.


In this work, we restrain the change-point detection problem to \emph{locally dependent} data, in which we are able to develop a general framework for high-dimensional data and non-Euclidean data objects. 
We leave the problem of long-range dependence for future studies.  The proposed framework builds upon an earlier work by \cite{chen2015graph}, in which the authors developed a nonparametric framework for change-point detection for generic data types under the assumption that the observations are \emph{independent}.  When there is local dependence in the data, the method in \cite{chen2015graph} could result in more false discoveries than it supposed to have (details see in Section \ref{sec:framework}).  To address this problem, we propose to use a new way of permutation -- circular block permutation with a random starting point.  This new way of permutation retains the local structure and could control type I error correctly when the sequence is locally dependent.  Moreover, simulation studies show that it retains the same level of power when the observations in the sequence are independent.


In the following, we use $\{\by_1, \by_2, \dots, \by_n\}$ to denote the data sequence, where $\by_t$ could lie in a high-dimensional or a non-Euclidean space.  We focus on the single change-point alternative to illustrate the idea, i.e., there possibly exists at most one change-point. 
The proposed procedure can be extended to the changed interval alternative and to multiple change-points.  Discussions on these extensions see in Section \ref{sec:multipleCP}.

The rest of the paper is organized as follows.  Section \ref{sec:review} briefly reviews the method introduced in \cite{chen2015graph} in utilizing a similarity graph constructed on observations for change-point detection.  You can skip reading this section if you are familiar with this method.   Section \ref{sec:framework} discusses the issues of the method in \cite{chen2015graph} when the sequence is locally dependent, and proposes a new permutation framework to address the issue.  Section \ref{sec:statistic} discusses more about the proposed test statistic and derives analytic formulas to calculate the test statistic.  Section \ref{sec:pvalue} derives the analytic $p$-value approximations for the test statistic and check how the approximations work for finite samples. We concludes the paper by discussing the choice of the block size in the new framework and the extension of the proposed framework to multiple change-points in Section \ref{sec:discussion}. 

\section{A brief review of \cite{chen2015graph}}
\label{sec:review}

In \cite{chen2015graph}, the authors assume that $\by_1, \by_2, \dots, \by_n$  are \emph{independent}.  They adapted the \emph{edge-count test} to the scan statistic framework: Each $t$ divides the observations into two samples, $\{\by_1,\dots,\by_t\}$ and $\{\by_{t+1},\dots,\by_n\}$, and the edge-count test is conducted to test whether these two samples are from the same distribution or not.  Then, the maximum of the scan statistics over $t$ is used as the test statistic. 

The edge-count test, introduced in \cite{friedman1979multivariate}, is a two-sample test that is based on a similarity graph constructed on the pooled observations of two samples.  The similarity graph can be a given graph that reflects the similarity between observations \citep{chen2013graph}.  More generally, it can be constructed based on a similarity measure through a certain criterion, such as a minimum spanning tree (MST)  \citep{friedman1979multivariate}, which is a tree connecting all observations with the total distance across edges minimized, a minimum distance pairing  \citep{rosenbaum2005exact}, or a nearest neighbor graph \citep{henze1988multivariate}.  

The edge-count test counts the number of edges in the graph that connect from different samples and reject the null of equal distribution when the count is significantly \emph{smaller} than its null expectation -- when the two samples are from the same distribution, the two samples are well mixed and this count is relatively large, so a small count is an evidence for rejecting the null of equal distribution. 
Let $G$ be the similarity graph on all observations in the sequence.  The edge-count test statistic at time $t$ is:

\vspace{-0.8em}
$$R_G(t) = \sum_{(i,j)\in G} I_{g_i(t)\neq g_j(t)}, \quad g_i(t) = I_{i>t}, $$
\vspace{-0.9em}

\noindent where $I_A$ is the indicator function that takes value 1 if event $A$ is true and 0 otherwise.  Figure \ref{fig:Rt} illustrates the computation of $R_G(t)$ on a small artificial data set (the observations are in 2-dimension for illustration purpose).

\begin{figure}[!htp]
\includegraphics[width=.244\textwidth]{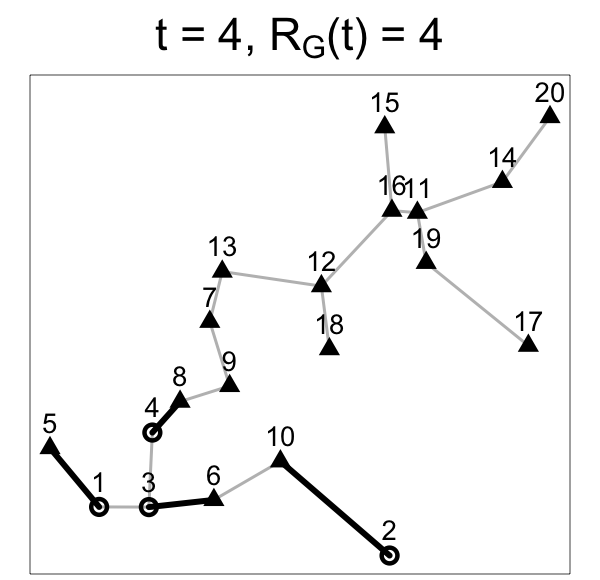} 
\includegraphics[width=.244\textwidth]{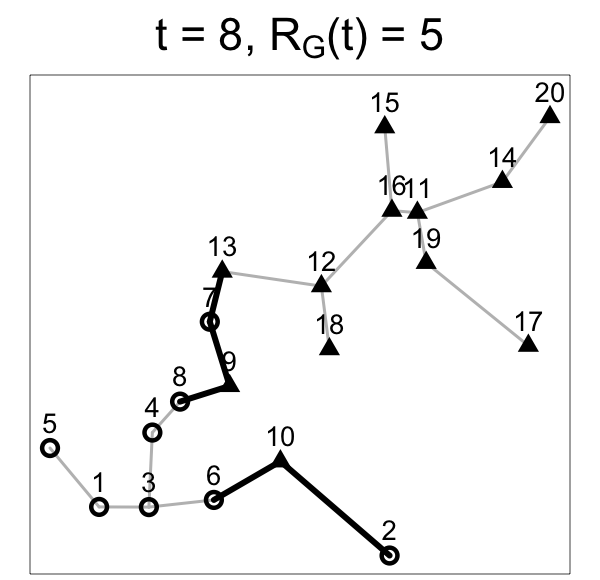}  
\includegraphics[width=.244\textwidth]{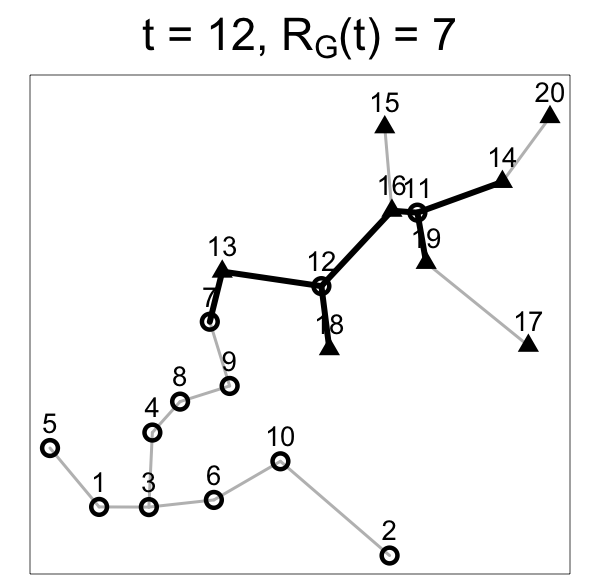} 
\includegraphics[width=.244\textwidth]{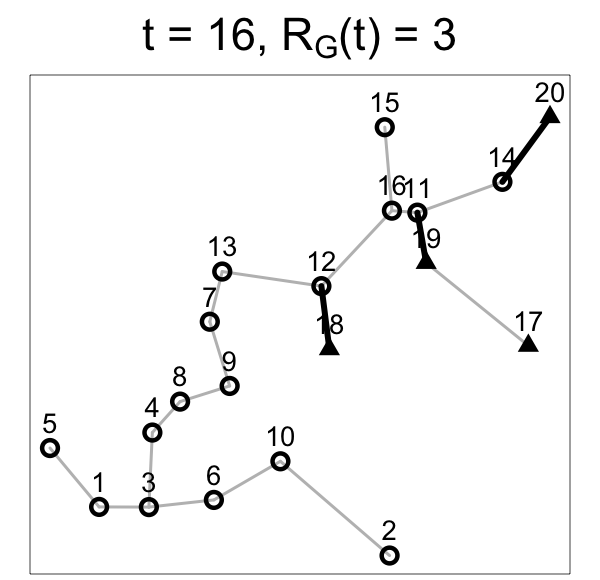}
\caption{The computation of $R_G(t)$ for 4 different $t$'s on a small artificial data set of length $n=20$ with $G$ be the MST on the Euclidean distance.  The index of each observation is beside each point.  The first 10 points are randomly drawn from $\mathcal{N}(\mathbf{0}, I_2)$ and the next 10 points are randomly drawn from $\mathcal{N}((2,2)^T, I_2)$.  Each $t$ divides the observations into two groups, one group for observations before and at $t$ (shown as circles) and the other group for observations after $t$ (shown as triangles).  Edges that connect observations from the two different groups are emboldened in the graph.  $G$ does not change as $t$ changes, but the group identities of some observations change, causing $R_G(t)$ to change. \label{fig:Rt}}
\vspace{-1em}
\end{figure}

Under the null hypothesis of no change-point and the assumption that $\by_i$'s are independent, the joint distribution of $\{\by_i:~i=1,\dots,n\}$ is the same under permutation.  Hence, the null distribution of $R_G(t)$ is defined to be the permutation distribution, which places $1/n!$ probability on each of the $n!$ permutations of $\{\by_i:~i=1,\dots,n\}$. 
 When there is no further specification, we denote by $\bP_\per$, $\bE_\per$, $\bV_\per$ probability, expectation, and variance, respectively, under the permutation null distribution.

The authors standardized $R_G(t)$ so that it is comparable across $t$.  Let
\begin{align} \label{eq:Zt}
  Z_G(t) & = -\frac{R_G(t)-\bE_\per(R_G(t))}{\sqrt{\bV_\per(R_G(t))}}.
\end{align}
The sign is flipped that a \emph{large} $Z_G(t)$ indicates a change-point.
The analytic expressions for $\bE_\per(R_G(t))$ and $\bV_\per(R_G(t))$ are given in \cite{chen2015graph}. 
%
%
The null hypothesis of no change-point is rejected if the scan statistic
\begin{equation}
  \label{eq:Zmax}
  \max_{n_0 \leq t \leq n_1} Z_G(t),  \quad (n_0, n_1 \text{ prespecified}) 
  \end{equation}
 is greater than a threshold.  When $n$ is small, this threshold could be determined by performing random permutations directly; when $n$ is large, \cite{chen2015graph} provided accurate analytic formulas to approximate the permutation $p$-value, allowing fast application of the method.  The authors also shown through simulations that this graph-based testing framework has better power than likelihood-based methods when the dimension of the data is moderate to high.

\section{A circular block permutation framework for locally dependent data}
\label{sec:framework}

The method in \cite{chen2015graph} assume that the observations are \emph{independent}.
Under the independence assumption, we can permute the order of the observations to get a pool of sequences that have the same distribution as the original sequence under the null hypothesis of no change-point.  However, when there is \emph{dependence structure} within the sequence, such as autocorrelation, the above argument no longer holds.  For example, if we apply the scan statistic \eqref{eq:Zmax} in \cite{chen2015graph} to an autocorrelated sequence and use the $p$-value calculated based on permutation, it rejects the null hypothesis more often than it should do (Figure \ref{fig:histPerm}, right panel).

\begin{figure}[!htp]

\includegraphics[width=.45\textwidth]{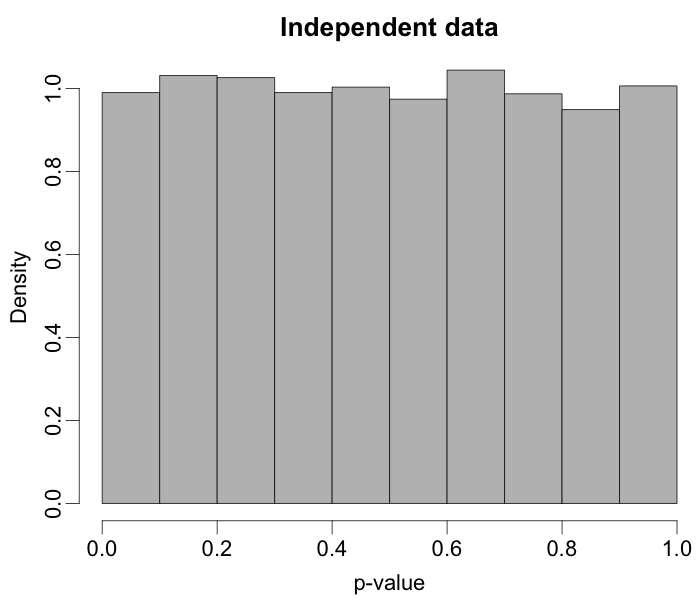} \ \  \
\includegraphics[width=.45\textwidth]{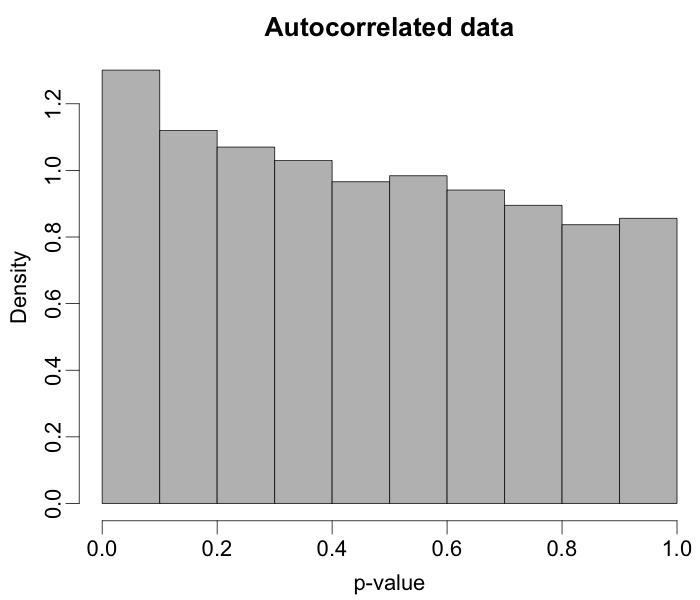}
\caption{Histograms of $p$-values using the method in \cite{chen2015graph} in testing homogeneity of 10,000 sequences of no change-point.  Left panel: the observations in each sequences are independently generated from multivariate normal distribution ($\by_t\overset{iid}{\sim}\mathcal{N}(\mathbf{0}, \Sigma),~d=10,~\Sigma(i,j)=|i-j|^{0.6},~n=200$).  Right panel: each sequence is generated from the multivariate autoregression model ($\by_t = \rho\by_{t-1} + \bvareps_t,~t=1,\dots,n$, with  $\by_0\sim \mathcal{N}(\mathbf{0},\frac{1}{1-\rho^2}\Sigma),~\bvareps_1,\dots,\bvareps_n \overset{iid}{\sim}\mathcal{N}(\mathbf{0}, \Sigma),~\rho=0.1,~d=10,~n=200$).  } 
\label{fig:histPerm}
\end{figure}


When the sequence has dependency over time, the permutation null distribution is no longer a good surrogate to the true null distribution as permutation destroy the local structure.  If the dependency structure can be removed from the sequence, the remaining sequence with independent observations can be analyzed through the method in \cite{chen2015graph}.  This is, however, not realistic for many applications with high-dimensional/non-Euclidean data sequences.  To address the local dependence issue, we propose a new null distribution that serves as a better surrogate to the true null distribution than the permutation null.

%
\subsection{Circular block permutation with a random starting point}
\label{sec:cbp}

The block-resampling bootstrap was proposed by \cite{kunsch1989jackknife} and independently by \cite{liu1992moving} as a resampling procedure for weakly dependent stationary observations.  The idea is that the dependency structure is preserved within the blocks.  This was extended to circular block resampling bootstrap by wrapping the data around in a circle before blocking them \citep{politis1992circular}.  For change-point analysis, the block permutation with fixed blocks starting from the first observation and the circular block bootstrap were studied on dependent data for one-dimensional observations \citep{kirch2006resampling}.

In light of these studies, we propose to use \emph{circular block permutation with a random starting point} to generate a pool of sequences representing realizations from approximate distributions of the original sequence with local dependency under the null of no change.  The recipe with block size $L$ on a sequence of length $n$ is as follows:

\begin{enumerate}[(1)]
\item The starting point is chosen uniformly from the $n$ observations, which is denoted as $k_0$.  If $k_0>1$, the first $k_0-1$ observations are moved to the end of the sequence: $\{\by_{k_0}, \dots,\by_n,\by_1,\dots,\by_{k_0-1}\}$.
\item The new sequence is blocked into $[n/L]$\footnote{For a real value $s$, we use $[s]$ to denote the largest integer that is no larger than $s$.} blocks of size $L$ starting from the first observation $\by_{k_0}$.  It is possible that the last block has less than $L$ observations.
\item The $[n/L]$ blocks are permuted.
\end{enumerate}

For the above recipe, it is easy to see that the resulting sequence from the circular block permutation with a random starting point is one permutation of the original sequence -- each observation appears in the resulting sequence and appears only once.  Then, the similarity graph on the resulting sequence is the same as that on the original sequence, which makes theoretical analysis on this framework tractable.  In addition, the randomized starting point ensures that the probability of any observation $\by_i$ appears at any location $j$ in the resulting sequence is $1/n$, ensuring unbiasedness. 

To make the theoretical treatment more tractable, we work under the following variant of the framework: We first augment the sequence by $x$ ($0\leq x<L$) pseudo observations by adding them to the end of the sequence so that $n+x$ is divisible by $L$.  These augmented $x$ observations have no edge connected to any other observations.  Then the recipe described above is applied to this augmented sequence with $n+x$ observations.
%
This variant is the same as the original version when $n$ is divisible by $L$ and works similarly when not.  In this variant, all the blocks are of size $L$, so the theoretical treatments are much more tractable.  In the following, we work under this variant and short the framework as `circular block permutation' or `\CBP' for simplicity.  We also use $n$ to denote $n+x$ for simplicity.  We use $\bP_{\CBP}$, $\bE_{\CBP}$, $\bV_{\CBP}$ to denote probability, expectation, and variance, respectively, under this framework. 
We consider $L$ to be fixed for the rest of the paper.  A discussion on how to choose $L$ in a data driven way is in Section \ref{sec:Lchoice}.

\subsection{Performance under CBP}

The standardized edge-count statistic under the circular block permutation can be defined as:
\begin{equation}
  \label{eq:Zbp}
  Z_{G,\CBP}(t) = - \frac{R_G(t)-\bE_{\CBP}(R_G(t))}{\sqrt{\bV_{\CBP}(R_G(t))}}.
\end{equation}
The scan statistic is then defined as:
\begin{equation} 
  \label{eq:maxZbp}
  \max_{n_0\leq t\leq n_1} Z_{G,\CBP}(t).
\end{equation}
In the following, with no further specification, we set $n_0=[0.05n]$ and $n_1=n-n_0$.

Figure \ref{fig:histcbp} shows the histograms of $p$-values in testing the homogeneity of sequences when the sequence has no change-point.   For each sequence, the $p$-value of the test is obtained through doing 100,000 CBPs, respectively.   Now we see that the type I error is correctly controlled for sequences of autocorrelated data. 


\begin{figure}[!htp]

\includegraphics[width=.45\textwidth]{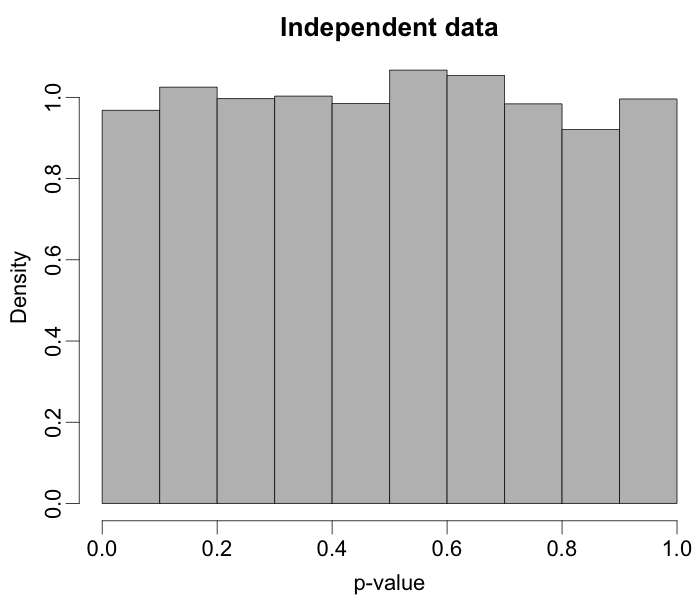} \ \  \
\includegraphics[width=.45\textwidth]{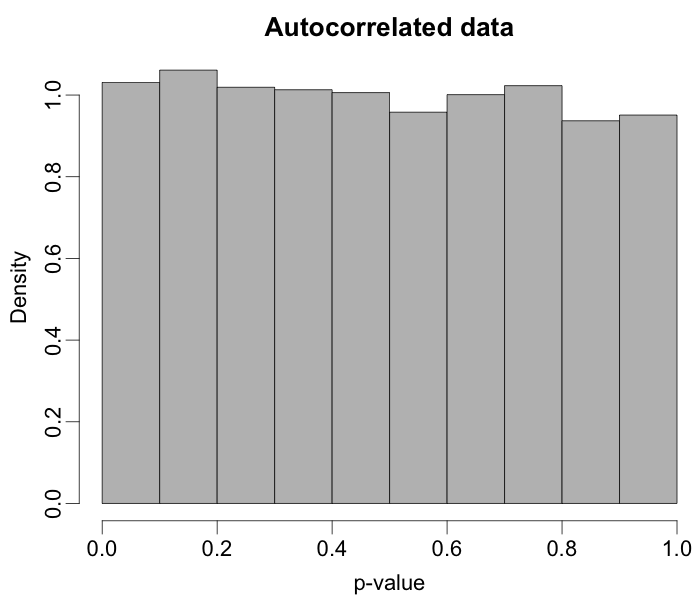}
\caption{Histograms of $p$-values under CBP (with block size $L=5$) in testing homogeneity of the same set of sequences in Figure \ref{fig:histPerm}.}
\label{fig:histcbp}
\end{figure}

%
%

An immediate follow-up question is whether the improvement in controlling the type I error of the circular block permutation framework come with a sacrifice on its power.  To get an idea of this issue, we compare the power of the two frameworks.
  Table \ref{tab:power} shows the estimated power based on 1,000 simulation runs.  In each simulation run, the sequence is generated in the same way as in Figure \ref{fig:histPerm}, while here, there is a mean shift in the middle of the sequence with the $L_2$ distance between the means before and after the change to be 2.  This specific alternative is chosen so that the tests have moderate power.  We see that when data are independent, the power under CBP is similar to that under permutation.   When the sequence is autocorrelated, the permutation framework cannot be used, while the CBP framework has power on par with the independent scenario. 
    

\begin{table}[!htp]
\caption{Estimated power based on 1,000 simulation runs.  Significance level set to be 0.05. \label{tab:power}}
\begin{tabular}{|c|c|c|}
\hline
& permutation & CBP, $L=5$ \\ \hline
Independent data & 0.791 & 0.779 \\ \hline
Autorrelated data & -- & 0.775 \\ \hline
\end{tabular}

\end{table}

%
%

These simulation results show the ability of CBP in controlling the type I error rate and at the same time keeping substantial power for sequences with local dependency.  In the above simulation runs, the expectation and variance of $R_G(t)$ under CBP, as well as the $p$-value of the test, are calculated by randomly sampling from the circular block permutation distribution.  This is very time consuming if one wants to get a good estimate of them.  In the following, we work on analytic expressions (or approximate analytic expressions when the exact analytic expression is hard to obtain) of these quantities to make this framework easy to use in practice.



\section{Analytic expressions under CBP}
\label{sec:statistic}

Set $m\equiv n/L$, there are in total $L\times m!$ CBPs and it is very time consuming to draw all these CBPs when $m$ is moderate to large.  In this section, we derive exact analytic expression for $\bE_{\CBP}(R_G(t))$ (Section \ref{sec:Ecbp}) and an approximate analytic expression for $\bV_{\CBP}(R_G(t))$ (Section \ref{sec:Vcbp}). 
In the following, for any scalar $s$, we use $(s)_+$ to denote $\max(0,s)$; and for set $S$, we use $|S|$ to denote the number of elements in the set.


\subsection{$\bE_{\CBP}(R_G(t))$}
\label{sec:Ecbp}

Let $\pi_{\CBP}(i)$ be the index of $\by_i$ under the circular block permutation.  Then
\begin{align}
\bE_{\CBP}(R_G(t)) & = \sum_{(i,j)\in G} P(g_{\pi_{\CBP}(i)}(t) \neq g_{\pi_{\CBP}(j)}(t)) \\
& = 2 \sum_{(i,j)\in G} P(\pi_{\CBP}(i)\leq t, \pi_{\CBP}(j) >t). \nonumber
\end{align}
The design of the circular block permutation ensures that, besides $t$, $n$ and the block size $L$, $P(\pi_{\CBP}(i)\leq t, \pi_{\CBP}(j) >t)$ only depends on $\delta_{ij}=\min(|i-j|, n-|i-j|)$, which is the smaller index difference between $\by_i$ and $\by_j$ in the circle formed by linking the end of the sequence to its start.  In particular,  the probability depends on which one of the following categories $\delta_{ij}$ belongs to: $\{\delta_{ij} =1\}$, $\dots$, $\{\delta_{ij} = L-1\}$, $\{\delta_{ij}\geq L\}$.  The probability is the same in each category and can be calculated exactly for each of them.  Hence, we can classify the edges in $G$ according to these categories and get Theorem \ref{thm:Ecbp}.

\begin{theorem}\label{thm:Ecbp}
 For each $t \in\{1,\dots,n\}$, we write $t$ in the form of $t=aL+b$ where $a=[t/L],~b=t-aL$, then
\begin{align*}
\bE_{\CBP}(R_G(t)) & = \sum_{k=1}^L 2\,p(k,a,b) |\mathcal{E}_k|,
\end{align*}
where
\begin{align*}
p(k,a,b) & = (\delta_{ij}-b)_+ \tfrac{a(m-a)}{n(m-1)} + (b-(L-k))_+ \tfrac{(a+1)(m-a-1)}{n(m-1)} \\
& \quad + (\min(b,L-k) -(b-k)_+)\tfrac{a(m-a-1)+(m-1)}{n(m-1)}, \\
\mathcal{E}_k & = \{(i,j)\in G: \delta_{ij} =k\},\quad  k=1,\dots, L-1, \\ 
\mathcal{E}_L & = \{(i,j)\in G: \delta_{ij} \geq L\}.
\end{align*}
\end{theorem}

\begin{remark}
When $t$ is divisible by $L$ ($t=aL$), we have $p(k,a,0) = \delta_{ij}\frac{a(m-a)}{n(m-1)}$. Then $\bE_{\CBP}(R_G(t)) = \frac{2a(m-a)}{n(m-1)}\sum_{k=1}^L k |\mathcal{E}_k|$.	
\end{remark}

\begin{proof}[Proof of Theorem \ref{thm:Ecbp}]
We compute the probability 
\begin{equation}\label{eq:pij}
P(\pi_{\CBP}(i)\leq t, \pi_{\CBP}(j)>t)
\end{equation}
under difference scenarios.

When $\delta_{ij}\geq L$, $\by_i$ and $\by_j$ are always in two different blocks.
If $b=0$, then $\pi_{\CBP}(i)\leq t$ only if the block containing $\by_i$ is placed in the first $a$ blocks after the circular block permutation, and $\pi_{\CBP}(j)>t$ only if the block containing $\by_j$ is placed in the last $m-a$ blocks after the circular block permutation.  So the probability \eqref{eq:pij} is 
$$\tfrac{a}{m} \tfrac{m-a}{m-1} = \tfrac{La(m-a)}{n(m-1)}.$$
If $b>0$, we need to discuss whether the block containing either $\by_i$ or $\by_j$ sits on $t$, whether $\by_i$ is in the first $b$ observations in the block, and whether $\by_j$ is in the last $L-b$ observations in the block.  Enumerating all possibilities, the probability \eqref{eq:pij} is
\begin{align*}
& \tfrac{a}{m}\left(\tfrac{m-a-1}{m-1} + \tfrac{1}{m-1}\tfrac{L-b}{L}  \right) + \tfrac{1}{m} \tfrac{b}{L} \tfrac{m-a-1}{m-1} = \tfrac{La(m-a)+b(m-2a-1)}{n(m-1)}.
\end{align*}

When $\delta_{ij}<L$, we also discuss the two scenarios: $b=0$ and $b>0$.


If $b=0$, $\by_i$ and $\by_j$ need to be in different blocks to have $(\pi_{\CBP}(i)\leq t, \pi(j)_{\CBP}>t)$.  Among the $L$ possible ways in blocking the sequence, $\delta_{ij}$ of them have $\by_i$ and $\by_j$ in different blocks, so the probability \eqref{eq:pij} is 
$$\tfrac{\delta_{ij}}{L}\tfrac{a(m-a)}{m(m-1)} = \tfrac{\delta_{ij} a(m-a)}{n(m-1)}.$$

If $b>0$, $\by_i$ and $\by_j$ could be in different blocks or in the same block to satisfy $(\pi_{\CBP}(i)\leq t, \pi_{\CBP}(j)>t)$ .  If they are in the same block, we denote that particular block by $B_1$.  If they are in different blocks, the two blocks much be adjacent.  Among the two blocks, we denote the left block to be $B_1$ (to make this argument consistent, the first block of the sequence and the last block of the sequence are considered to be adjacent and the last block of the sequence is considered to be on the left of the first block of the sequence).   We then let the adjacent block right of $B_1$ to be $B_2$.
For block $B_1$, we further divide it into two sub-regions with $B_{1,l}$ denoting the first $b$ location(s) of the block and $B_{1,r}$ denoting the rest $L-b$ location(s) of the block.  We define $B_{2,l}$ and $B_{2,r}$ similarly for block $B_2$.  


Then, there are four configurations for the placements of $i$ and $j$ for each of the two scenarios: (i) $i$ on the left of $j$ within $B_1\cup B_2$, and (ii) $i$ on the right of $j$ within $B_1\cup B_2$.  The four configurations are listed in Tables \ref{tab:ptable1} and \ref{tab:ptable2} for these two scenarios, respectively.  Together in the tables are the probability of having each of the configuration out of $L$ different ways of doing the blocking (Prob. 1 in the tables) and the proportion of the permutations in terms of permuting the blocks so that $(\pi_{\CBP}(i)\leq t, \pi_{\CBP}(j)>t)$ given the configuration (Prob. 2 in the tables).  For each of the two scenarios, summing over the product of the two probabilities (Prob. 1 and Prob. 2 in the tables) gives \eqref{eq:pij}.

\begin{table}[h]
\caption{Four configurations of the placement of $i$ and $j$ when $\delta_{ij}<L$, $b>0$ and $i$ on the left of $j$ within $B_1\cup B_2$.  For each configuration, ``Prob. 1" is the probability of having the configuration out of $L$ different ways of doing the blocking, and ``Prob. 2" is the proportion of the permutations in terms of permuting the blocks so that $(\pi_{\CBP}(i)\leq t, \pi_{\CBP}(j)>t)$ given the configuration. (In this table, $\delta_{ij}$ is shortened as $\delta$ to avoid cumbersome.) \label{tab:ptable1}}
\begin{tabular}{|c|c|c|c||c|c|}
\hline
$B_{1,l}$ & $B_{1,r}$ & $B_{2,l}$ & $B_{2,r}$ & Prob. 1 & Prob. 2 \\ \hline
$i$ & $j$ & & & $\tfrac{\min(b,L-\delta) -(b-\delta)+}{L}$ & $\tfrac{1}{m}$ \\ \hline
$i$ & & $j$ & & $\tfrac{(b-(L-\delta))_+}{L}$ & $\tfrac{(a+1)(m-a-1)}{m(m-1)}$ \\ \hline
& $i$ & $j$ & & $\tfrac{\min(b,L-\delta) -(b-\delta)+}{L}$ & $\tfrac{a(m-a-1)}{m(m-1)}$ \\ \hline
& $i$ & & $j$ &  $\tfrac{(\delta-b)_+}{L}$ & $\tfrac{a(m-a)}{m(m-1)}$ \\ \hline
\end{tabular} 
\end{table}

\begin{table}[h]
\caption{Four configurations of the placement of $i$ and $j$ when $\delta_{ij}<L$, $b>0$ and $i$ on the right of $j$ within $B_1\cup B_2$.  Other notations follow Table \ref{tab:ptable1}. \label{tab:ptable2}}
\begin{tabular}{|c|c|c|c||c|c|}
\hline
$B_{1,l}$ & $B_{1,r}$ & $B_{2,l}$ & $B_{2,r}$ & Prob. 1 & Prob. 2 \\ \hline
$j$ & $i$ & & & $\tfrac{\min(b,L-\delta) -(b-\delta)+}{L}$ & 0 \\ \hline
$j$ & & $i$ & & $\tfrac{(b-(L-\delta))_+}{L}$ & $\tfrac{(a+1)(m-a-1)}{m(m-1)}$ \\ \hline
& $j$ & $i$ & & $\tfrac{\min(b,L-\delta) -(b-\delta)+}{L}$ & $\tfrac{a(m-a) + (m-a-1)}{m(m-1)}$ \\ \hline
& $j$ & & $i$ &  $\tfrac{(\delta-b)_+}{L}$ & $\tfrac{a(m-a)}{m(m-1)}$ \\ \hline
\end{tabular}
\end{table} 

Since 
$\tfrac{1}{m}+\tfrac{a(m-a-1)}{m(m-1)} = \tfrac{a(m-a) + (m-a-1)}{m(m-1)},$
both summations give, for $\delta_{ij}<L$,
\begin{align*}
& (\min(b,L-\delta_{ij}) -(b-\delta_{ij})+)\tfrac{a(m-a-1)+(m-1)}{n(m-1)} \\
& + (b-(L-\delta_{ij}))_+ \tfrac{(a+1)(m-a-1)}{n(m-1)} + (\delta_{ij}-b)_+ \tfrac{a(m-a)}{n(m-1)}. 
\end{align*}
Then, the theorem follows as the result for $\delta_{ij}\geq L$ is a special case of the above expression with $\delta_{ij}$ replaced by $L$.
\end{proof}


\subsection{$\bV_{\CBP}(R_G(t))$}
\label{sec:Vcbp}

For variance, we need to figure out
$\bE_{\CBP}(R_G^2(t)) = \sum_{(i,j), (u,v)\in G} P(g_{\pi_{\CBP}(i)}(t)\neq g_{\pi_{\CBP}(j)}(t), g_{\pi_{\CBP}(u)}(t) \neq \pi_{\pi_{\CBP}(v)}(t)).$
Then, $\bV_{\CBP}(R_G(t))$ follows as $\bE_{\CBP}(R_G^2(t))-(\bE_{\CBP}(R_G(t)))^2$ with the analytic expression for $\bE_{\CBP}(R_G(t))$ provided in Section \ref{sec:Ecbp}.

When $i,j,u,v$ are all different, we need to consider $\binom{4}{2}=6$ index-pairs and whether they could be within a block or not.  This is much more complicated than the calculation in $\bE_{\CBP}(R_G(t))$ where only one index-pair is considered.  In the derivation of $\bE_{\CBP}(R_G(t))$, when $t$ is divisible by $L$ ($b=0$ in the proof), the derivation is much easier.
Therefore, we work out the exact analytic expression for $\bV_{\CBP}(R_G(t))$ for $t$ divisible by $L$, which is already very complicated, and do extrapolations for other $t$'s.  We then compare the result obtained in this way with doing circular block permutation directly.   The following theorem gives exact analytic expression for $\bV_{\CBP}(R_G(t))$ when $t=aL, a=0,\dots,m$.  

\begin{theorem}\label{thm:Vcbp}
For $t=aL, a=1,\dots,m$, we have
\begin{align*}
\bV_{\CBP}(R_G(t)) & = c_1 p_1(a) + c_2 p_2(a) + c_3 p_3(a) - c_0^2 p_1^2(a), 
\end{align*}
where 
\allowdisplaybreaks
\small
\begin{align*}
 p_1&(a)  = \tfrac{2a(m-a)}{m(m-1)}, \\
 p_2&(a)  = \tfrac{1}{2}p_1(a),\\
p_3&(a)  = \tfrac{4a(a-1)(m-a)(m-a-1)}{m(m-1)(m-2)(m-3)}, \\
 c_0 & = \tfrac{1}{L}\sum_{k=1}^L k |\mathcal{E}_k|, \\
 c_1&  = \tfrac{1}{L}\sum_{k=1}^L k |\mathcal{E}_k| \\
& \ + \tfrac{1}{L}\sum_{\scriptsize (i,j),(i,u)\in G;\ j\neq u}\{ (L-\delta_{ju}) I_{h_0(i,j,u)<3,\delta_{ju}<L}  \\
& \hspace{34mm} + \min(\delta_{ij},\delta_{iu}) I_{h_0(i,j,u)=3} I_{\max(\delta_{ij},\delta_{iu},\delta_{ju}) \neq \delta_{ju}} \} \\
& \ + \tfrac{1}{L} \sum_{\tiny\begin{array}{c}(i,j),(u,v)\in G \\ i,j,u,v\text{ all different} \end{array}} \{ I_{h_1(i,j,u,v)=2}((L-\delta_{iu}-\delta_{jv}+x(iu,jv)) I_{\delta_{iu}<L, \delta_{jv}<L} \\
& \hspace{55mm} + (L-\delta_{iv}-\delta_{ju}+x(iv,ju)) I_{\delta_{iv}<L, \delta_{ju}<L}) \\
& \hspace{22mm} + I_{h_1(i,j,u,v)=3} (2L-\delta_{\min,3})_+  \\
& \hspace{29mm} \times((1-O_3)(I_{\delta_{ij}<L,\delta_{uv}\geq L}+I_{\delta_{ij}\geq L,\delta_{uv}<L}) + I_{\delta_{ij}\geq L, \delta_{uv}\geq L}) \\
& \hspace{22mm}+ I_{h_1(i,j,u,v)=4} ((L+\delta_{uv}-\max(\delta_{iu},\delta_{iv},\delta_{ju},\delta_{jv}) I_{\delta_{ij}\geq L}) \\
& \hspace{29mm} +(L+\delta_{ij}-\max(\delta_{iu},\delta_{iv},\delta_{ju},\delta_{jv}))  I_{\delta_{uv}\geq L}) \\
& \hspace{22mm} + I_{h_1(i,j,u,v)\geq 5} (\delta_{uv} I_{\delta_{\max}(i,j,u,v)=\delta_{ij}} + \delta_{ij} I_{\delta_{\max}(i,j,u,v)=\delta_{uv}} \\
& \hspace{29mm} + \delta_{jv} I_{\delta_{\max}(i,j,u,v)=\delta_{iu}} I_{\delta_{ij}=\delta_{iv}+\delta_{jv}}  \\
& \hspace{29mm} + \delta_{ju} I_{\delta_{\max}(i,j,u,v)=\delta_{iv}} I_{\delta_{ij}=\delta_{iu}+\delta_{ju}}  \\
& \hspace{29mm} + \delta_{iv}  I_{\delta_{\max}(i,j,u,v)=\delta_{ju}} I_{\delta_{ij}=\delta_{iv}+\delta_{jv}} \\
& \hspace{29mm} + \delta_{iu} I_{\delta_{\max}(i,j,u,v)=\delta_{jv}} I_{\delta_{ij}=\delta_{iu}+\delta_{ju}} ) \}\\
c_2 & = \tfrac{1}{L} \sum_{\scriptsize (i,j),(i,u)\in G;\ j\neq u} \{ L I_{h_0(i,j,u)=0} + \min(\delta_{ij},\delta_{iu},\delta_{ju}) I_{h_0(i,j,u)=1} \\
& \hspace{33mm} + (\max(\delta_{ij},\delta_{iu},\delta_{ju})-L) I_{h_0(i,j,u)=2} \} \\
& \ + \tfrac{1}{L}\sum_{\tiny \begin{array}{c}(i,j),(u,v)\in G \\ i,j,u,v\text{ all different} \end{array}} \{ I_{h_1(i,j,u,v)=1} (L-\delta_{\min}(i,j,u,v)) I_{\delta_{ij}\geq L} I_{\delta_{uv}\geq L}   \\
& \hspace{22mm}  + I_{h_1(i,j,u,v)=2}((L-\delta_{iu})_+ + (L-\delta_{iv})_+ + (L-\delta_{ju})_+ + (L-\delta_{jv})_+ \\
& \hspace{29mm} + (2\delta_{iu}+2\delta_{jv}-2L-2x(iu,jv)) I_{\delta_{iu}<L, \delta_{jv}<L} \\
& \hspace{29mm} + (2\delta_{iv}+2\delta_{ju}-2L-2x(iv,ju)) I_{\delta_{iv}<L, \delta_{ju}<L} ) \\
& \hspace{22mm}  + I_{h_1(i,j,u,v)=3}( (L-\min(\delta_{iu},\delta_{iv},\delta_{ju},\delta_{jv})) I_{\delta_{ij}<L,\delta_{uv}<L} \\
& \hspace{29mm} + I_{\delta_{ij}<L,\delta_{uv}\geq L}(\delta_{ij}- (1-O_3)|\delta_{uv}-2L|) \\
& \hspace{29mm} + I_{\delta_{ij}\geq L,\delta_{uv}< L}(\delta_{uv}-(1-O_3)|\delta_{ij}-2L|) \\
& \hspace{29mm} + I_{\delta_{ij}\geq L,\delta_{uv}\geq L}(\delta_{\min,3}(i,j,u,v)-L-2(\delta_{\min,3}(i,j,u,v)-2L)_+ \\
& \hspace{29mm} + I_{\delta_{jv}\geq L}\delta_{iu} + I_{\delta_{ju}\geq L}\delta_{iv} + I_{\delta_{iu}\geq L}\delta_{ju} + I_{\delta_{iv}\geq L}\delta_{ju} ) )\\
& \hspace{22mm} + I_{h_1(i,j,u,v)=4}((\max(\delta_{iu},\delta_{iv},\delta_{ju},\delta_{jv})-L) (I_{\delta_{ij}\geq L} + I_{\delta_{uv}\geq L}) \\
& \hspace{29mm} + \delta_{uv} (I_{\delta_{iu},\delta_{iv}\geq L} + I_{\delta_{ju},\delta_{jv}\geq L}  ) \\
& \hspace{29mm} + \delta_{ij} (I_{\delta_{iu},\delta_{ju}\geq L} +I_{\delta_{iv},\delta_{jv}\geq L}  )) \\
& \hspace{22mm} + I_{h_1(i,j,u,v)=5}(\max(\delta_{iu},\delta_{iv},\delta_{ju},\delta_{jv})-L)_+ \} \\
c_3 & = \tfrac{1}{L} \sum_{\tiny \begin{array}{c}(i,j),(u,v)\in G \\ i,j,u,v\text{ all different} 
\end{array}} \{ L I_{h_1(i,j,u,v)=0}  + \delta_{\min}(i,j,u,v)  I_{h_1(i,j,u,v)=1}  \\
& \hspace{22mm} + I_{h_1(i,j,u,v)=2}(\delta_{\min,2}(i,j,u,v) -L \\
& \hspace{29mm}  + I_{\delta_{ij}<L, \delta_{uv}<L} (x_{ij,uv}-\delta_{\min,2}(i,j,u,v) + L) \\
& \hspace{29mm}  + I_{\delta_{iu}<L, \delta_{jv}<L} (x_{iu,jv}-\delta_{\min,2}(i,j,u,v) + L) \\
& \hspace{29mm} + I_{\delta_{iv}<L,\delta_{ju}<L} (x_{iv,ju} -\delta_{\min,2}(i,j,u,v) + L)) \\
& \hspace{22mm} + I_{h_1(i,j,u,v)=3}(\delta_{\min,3}-2L)_+ \}
\end{align*}
with
\begin{align*}
h_0(i,j,u) & = I_{\delta_{ij}<L} + I_{\delta_{iu}<L}+ I_{\delta_{ju}<L}, \\
h_1(i,j,u,v) &= I_{\delta_{ij}<L} + I_{\delta_{iu}<L} + I_{\delta_{iv}<L} + I_{\delta_{ju}<L} + I_{\delta_{jv}<L} + I_{\delta_{uv}<L},\\
\delta_{\max}(i,j,u,v) & = \max(\delta_{ij},\delta_{iu},\delta_{iv},\delta_{ju},\delta_{jv},\delta_{uv}),\\
\delta_{\min}(i,j,u,v) & = \min(\delta_{ij},\delta_{iu},\delta_{iv},\delta_{ju},\delta_{jv},\delta_{uv}),\\
\delta_{\min,2}(i,j,u,v) & = \text{sum of the two smallest values among }\{\delta_{ij},\delta_{iu},\delta_{iv},\delta_{ju},\delta_{jv},\delta_{uv}\},\\
\delta_{\min,3}(i,j,u,v) & = \text{sum of the three smallest values among }\{\delta_{ij},\delta_{iu},\delta_{iv},\delta_{ju},\delta_{jv},\delta_{uv}\},\\
s(i,j) & = \left\{\begin{array}{ll} \min(i,j) & \text{ if $|i-j|<L$}, \\ \max(i,j) & \text{ if $n-|i-j|<L$}. \end{array}  \right. \\
\delta_{ij,uv} & = s(u,v)-s(i,j), \\
b_{ij,uv} & = \delta_{ij,uv} \text{ mod } L, \\
x_{ij,uv} & = (\min(\delta_{ij}, b_{ij,uv}+\delta_{uv})-b_{ij,uv})_+ + (\min(\delta_{ij},b_{ij,uv}+\delta_{uv}-L))_+, \\
O_3 & = I_{\delta_{ij},\delta_{iu},\delta_{ju}<L} + I_{\delta_{ij},\delta_{iv},\delta_{jv}<L} + I_{\delta_{uv},\delta_{iu},\delta_{iv}<L} + I_{\delta_{uv},\delta_{ju},\delta_{jv}<L}.
\end{align*}
\end{theorem}

\normalsize
The complete proof of this Theorem is in Appendix \ref{sec:Vcbp_Proof}. 

\begin{remark}\label{remark:coef}
For $c_0$, a high level explanation is that it is the average of the number of edges whose end nodes appear in different blocks over all $L$ possible ways of blocking\footnote{Based on the recipe of the circular block permutation, the blocks resulted from a random starting point at $\by_t$ are the same as the blocks resulted from a random starting point at $\by_{t+L}$.  Hence, there are only $L$ different ways of blocking.}.  Let $\omega$ represents one of such blockings and $\Omega$ be the set of all $L$ ways of blockings, then $c_0 = \frac{1}{L} \sum_{\omega\in\Omega} c_0(\omega)$.  The other three coefficients $c_1$, $c_2$, and $c_3$, involve two edges.  The two edges do not need to be distinct, i.e., they could degenerate to be the same edge; or the two edges could share the same node.  Then, under a certain blocking $\omega$, $c_1(\omega)$ is the number of pairs of edges whose end nodes only appear in two distinct blocks with both edges having their end nodes appearing in different blocks, $c_2(\omega)$ is the number of pairs of edges whose end nodes appear in three distinct blocks with both edges having their end nodes appearing in different blocks, and $c_3(\omega)$ is the number of pairs of edges whose end nodes appear in four distinct blocks.  It is not hard to see that $c_1(\omega) + c_2(\omega) + c_3(\omega) = c_0^2(\omega)$.  Then, 
$$c_1 + c_2 + c_3 = \tfrac{1}{L} \sum_{\omega\in\Omega} c_0^2(\omega) \geq \left(\tfrac{1}{L} \sum_{\omega\in\Omega} c_0(\omega)\right)^2 = c_0^2.$$

When $L=1$, the equality always holds.  Indeed, when $L=1$, the coefficients can be simplified to be $c_0=c_1=|G|$, $c_2 = \sum_{i=1}^n |G_i|^2 - 2|G|$ and $c_3 = |G|^2 -\sum_{i=1}^n |G_i|^2 + |G|$, where $|G|$ is the number of edges in the graph $G$, and $G_i$ is the subgraph in $G$ that connect to node $\by_i$.  So $|G_i|$ is the degree of node $\by_i$.  It is clear that $c_1+c_2+c_3 = |G|^2 = c_0^2$.  However, for $L>1$, the equality usually does not hold unless under very special cases that $c_0(\omega)$ is the same for all $\omega\in \Omega$.
\end{remark}

In Theorem \ref{thm:Vcbp}, we provide the exact analytic expression for $\Var_\CBP(R(t))$ when $t$ is divisible by $L$.  Unless $L=1$ that the coefficients, $c_1$, $c_2$ and $c_3$, could be greatly simplified, the expressions for these coefficients are in general fairly complicated for $L>1$.  It can be imagined that the exact analytic expression for $\Var_\CBP(R(t))$ when $t$ is not divisible by $L$ would be much more complicated.   In addition, the computation time for these coefficients when $L>1$ is not negligible as it needs to do some complicated counting.  If there are special structures of the graph, the expression might be simplified using the structures.  While for a general graph $G$, the analytic expression could not be further simplified as two edges are involved in computing the variance and we need to consider all the possible combinations of whether the $\binom{4}{2}=6$ pairwise index differences are smaller than $L$ or not when the four end points of the two edges are distinct.   In a typical run for a 1,000-length sequence with local dependence, the computation time for getting all the coefficients with $L=5$ is about 8 second on the 12-inch MacBook (2015), which is acceptable; while the not-derived analytic expressions for the variance under CBP at $t=aL+b, 0<b<L$ would be much more complicated (the magnitude of the complication different can be inferred from the exact analytic expression of $\bE_\CBP(R(t))$ in Section \ref{sec:Ecbp} under $b=0$ and $b>0$) and require much more time to compute.  Combining all the factors, we propose to fill-in $\bV_{\CBP}(R_G(t))$ at $t$ not divisible by $L$ by extrapolating from the values at $t=aL, a=0,1,\dots,m$.

\begin{figure}[h]  
\includegraphics[width=0.48\textwidth]{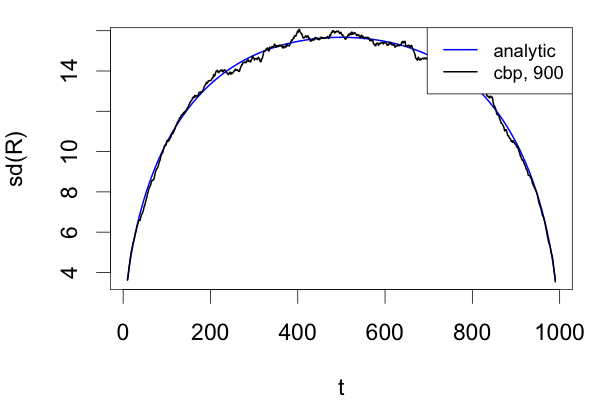} 
\includegraphics[width=.48\textwidth]{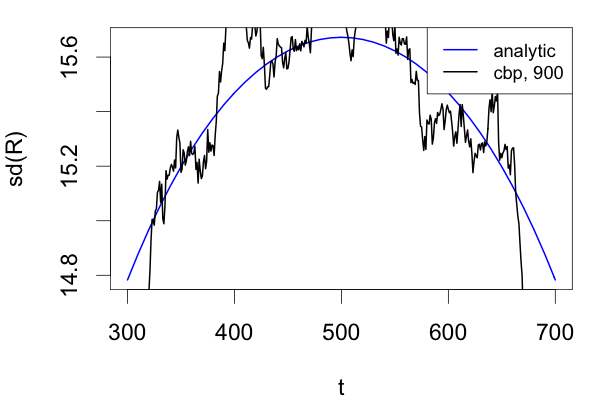}
\includegraphics[width=0.48\textwidth]{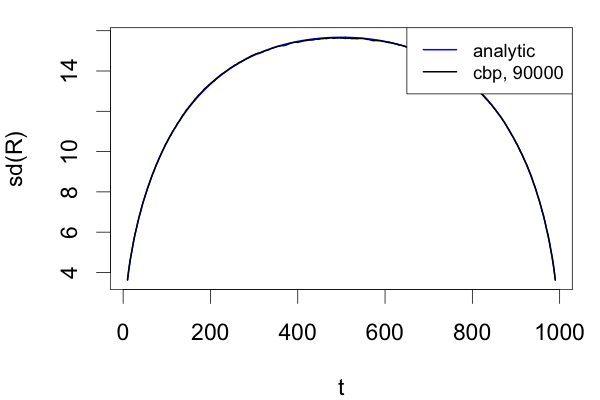}
\includegraphics[width=0.48\textwidth]{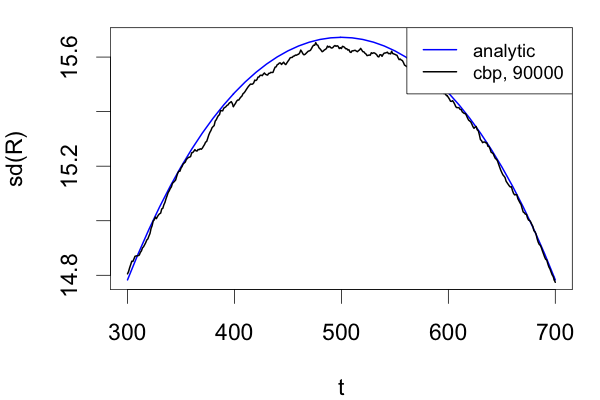}
\caption{Standard deviation of $R_G(t)$ under circular block permutation from the analytic expression given in Theorem \ref{thm:Vcbp} (the values at $t=aL+b, 0<b<L$ are extrapolated) and from doing circular block permutation directly.  The left panels are plotting the whole line and the right panels are plotting the middle part of the sequence.  In each plot, the blue line is based on the analytic formula and the black line is based on directly doing circular block permutation with 900 CBPs for the top panels and 90,000 CBPs for the bottom panels. \label{fig:sdR}}
\end{figure}

Figure \ref{fig:sdR} shows the standard deviation of $R_G(t)$ under circular block permutation ($\text{SD}_{\CBP}(R_G(t))$) for a 1000-length sequence with local dependence.  The blue line in each plot is based on the analytic formula provided in Theorem \ref{thm:Vcbp} with the values at $t=aL+b, 0<b<L$ filled in by extrapolation, and the black line in each plot is based on 900 CBPs (top panels) and 90,000 CBPs (bottom panels).  Here, 900 CBPs were chosen as it uses a similar amount of time to that by computing the standard deviation based on Theorem \ref{thm:Vcbp} and extrapolation for this sequence.  We can see clearly from the top panels that 900 CBPs is not enough as the results are fluctuating widely.  It is important to have a good estimate of $\text{SD}_{\CBP}(R_G(t))$ as it standardize the raw statistic $R_G(t)$ and a bad estimate could lead to a bad estimate of the change-point location.  When we increase the number of CBPs to 90,000 (using 100-fold times as for getting the analytic results), the values based on CBPs directly is much better (bottom left panel).  However, if we zoom into the middle part of the sequence, we could still see the black line wiggling around, which could cause inaccurate estimate of the change-point location.  This toy example shows that we would need even more number of CBPs to get an estimate that is as good as those from the analytic formula with extrapolation.  

Therefore, we recommend to use the analytic formula given in Theorem \ref{thm:Vcbp} to get exact values at $t=aL, a=0,1,\dots,m$ and use extrapolation to get values at $t=aL+b,0<b<L$.  This approach gives us accurate enough estimate for $\bV_{\CBP}(R_G(t))$ with a reasonable fast enough computing time.  In the following, $Z_{G,\CBP}(t)$ is defined with $\bV_{\CBP}(R_G(t))$ computed in this recommended way.

\section{Analytic $p$-value approximations}
\label{sec:pvalue}
Now, we have a relative fast analytic way to compute the standardized statistic $Z_{G,\CBP}(t)$.  The next question is how large the scan statistic 
$$\max_{n_0\leq t\leq n_1} Z_{G,\CBP}(t)$$
needs to be to constitute sufficient evidence against the null hypothesis of homogeneity, i.e., we are concerned with the tail probability of the scan statistics under $H_0$:
\begin{align}\label{eq:tailp}
\bP(\max_{n_0\leq t\leq n_1} Z_{G,\CBP}(t)>b).
\end{align}

To obtain this tail probability, we can directly perform the circular block permutation, which would be time consuming in obtaining a reasonably accurate estimate.  Therefore, we seek to derive an analytic expression for this tail probability.  In the rest of this section, we first study the asymptotic distribution of the process $\{Z_{G,\CBP}(t)\}$.  We derive approximate analytic expression for the tail probability for the limiting process and then refine the approximation to work for finite $n$.

\subsection{Asymptotic property of the process}  
\label{sec:asymp}

Here, we derive the limiting distribution of $\{Z_{G,\CBP}([mw]L): \epsilon \leq w\leq 1-\epsilon \}$ for any $0<\epsilon <0.5$ under circular block permutation.  
We first introduce some notations.  We write the nodes connected by an edge $e$ be $e_-$ and $e_+$ with $e_-<e_+$. 
For node $i$ and edge $e=(e_-,e_+)$, let
\allowdisplaybreaks
\begin{align*}
A_{e,L,0} & = \{e^*: \min(\delta(e^*_-, e_-), \delta(e^*_-,e_+), \delta(e^*+,e_-), \delta(e^*_+,e_+))<L\}, \\
A_{e,L,1} & = A_{e,L,0} \cup\left(\bigcup_{\{e^\prime: e^\prime \in G_{e^*_-}\cup G_{e^*_+}, \forall  e^*\in A_{e,L,0}\}} A_{e^\prime,L,0}\right), \\
A_{e,L,2} & = A_{e,L,1} \cup \left(\bigcup_{\{e^\prime: e^\prime \in G_{e^*_-}\cup G_{e^*_+}, \forall e^*\in A_{e,L,1}\}} A_{e^\prime,L,1}\right), \\
A_{i,L,0} & = \{e^*: \min(\delta(e^*_-,i), \delta(e^*_+,i))<L\}, \\
A_{i,L,1} & = A_{i,L,0} \cup \left(\bigcup_{\{e^\prime:e^\prime \in G_{e^*_-}\cup G_{e^*_+}, \forall e^*\in A_{e,L,0}\}} A_{e^\prime,L,0} \right), \\
A_{i,L,2} & = A_{i,L,1} \cup \left( \bigcup_{\{e^\prime: e^\prime \in G_{e^*_-}\cup G_{e^*_+}, \forall e^*\in A_{e,L,1} \}} A_{e^\prime,L,1}\right).
\end{align*}
Here, $A_{e,L,0}$ is the set of edges whose end nodes could be within the same block with any of the end nodes in $e$ under some circular block permutations.  This can be viewed as the neighbors of $e$.  Then, $A_{e,L,1}$ is the set of edges whose end nodes could be within the same block with any of the end nodes of the edges in $A_{e,L,0}$, so $A_{e,L,1}$ can be viewed as the set containing neighbors of $A_{e,L,0}$.  Similarly, $A_{e,L,2}$ can be viewed as the set containing neighbors of $A_{e,L,1}$.  The other three sets, $A_{i,L,0}$, $A_{i,L,1}$ and $A_{i,L,2}$, are defined similarly but initiated from a node $i$.

For a certain block $i$, let $D_i$ be the number of edges in $G$ that connect a node in block $i$ to another node not in block $i$ under the blocking that block $i$ exists.  Let $\bE_\Omega$ be the expectation that places probability $\frac{1}{L}$ on each $\omega\in\Omega$ with $\Omega$ defined in Remark \ref{remark:coef}.  In the following, we write $a_n = O(b_n)$ when $a_n$ has the same order as $b_n$, and write $a_n=o(b_n)$ when $a_n$ has order smaller than $b_n$.
The limiting distribution of the stochastic process needs the following conditions.
\begin{condition}
	$|G|=O(n^\alpha), 1\leq \alpha< \tfrac{8}{7}$; $\sum_{e\in G} |A_{e,L,1}||A_{e,L,2}| = o(n|G|^{\frac{1}{2}})$; $\sum_{i=1}^n |A_{i,L,1}| |A_{i,L,2}| = o(n^{1.5})$.
\end{condition}
\begin{condition}
	$\bE_\Omega(\sum_{i=1}^m D_i^2) - \frac{1}{m} (\bE_\Omega(\sum_{i=1}^m D_i))^2 = O(\bE_\Omega(\sum_{i=1}^m D_i^2))$.
\end{condition}
\begin{condition}
	$\bE_\Omega(\sum_{i=1}^m D_i^2) = O(|G|)$.
\end{condition}

\begin{theorem}\label{thm:asym}
Under Conditions 1 and 2, or under Conditions 1 and 3,  as $n\rightarrow\infty$, $\forall \epsilon\in(0,0.5)$, $\{Z_{\CBP}([mw]L): \epsilon < w < 1-\epsilon \} $ converges in finite dimensional distributions to a Gaussian process, which we denote as $\{Z^\star_{\CBP}(w): \epsilon < w < 1-\epsilon \}$.
\end{theorem}

The complete proof of the theorem is in Appendix \ref{sec:asym_Proof}.  It utilizes a proof technique that is similar to that in \cite{chen2015graph}.

\begin{remark}
	Condition 1 under $L=1$ are relaxed versions of the conditions for the limiting distribution under permutation null distribution in \cite{chen2015graph}\footnote{We do not consider $|G|=O(n^\alpha), 0<\alpha<1$ here as such a graph only makes use of a tiny portion of the observations and is not effective in practice.}.  For Conditions 2 and 3, we only need one of them to hold.  Condition 2 says that the graph shall not be flat: $D_i$'s shall not be of the same order across $i$'s.  However, when $|G|=O(n)$, such a flat graph would also be acceptable as Condition 3 holds.
\end{remark}

%

\subsection{Analytic formulas} \label{sec:pval_formula}
We now examine the asymptotic behavior of the tail probability \eqref{eq:tailp}.  Our approximation require the function $\nu(x)$ defined as 
$\nu(x) = 2x^{-2}\exp\left\{-2\sum_{i=1}^\infty \tfrac{1}{i}\Phi\left(-\tfrac{1}{2}x\sqrt{i}\right)\right\}, x>0$
This function is closely related to the Laplace transform of the overshoot over the boundary of a random walk.  A simple approximation given in \cite{siegmund2007statistics} is sufficient for numerical purpose:
$\nu(x) \approx \tfrac{(2/x)(\Phi(x/2)-0.5)}{(x/2)\Phi(x/2)+\phi(x/2)}.$

Based on Theorem \ref{thm:asym}, under Conditions 1 and 2 (or 3), following similar arguments in \cite{chen2015graph}, we have as $n\rightarrow\infty$, for $b=O(\sqrt{n})$, $n_0, n_1 = O(n)$,
\begin{align}\label{eq:pval1} 
 P(\max_{n_0\leq t\leq n_1} Z_{\CBP}(t)>b)
 \sim b\phi(b) \sum_{n_0\leq t\leq n_1} C(t) \nu\left(\sqrt{2b^2 C(t) } \right),
\end{align}
where $C(t) =\tfrac{1}{L}\left. \tfrac{\partial \rho(s,t)}{\partial s}\right|_{s=t},$
with $\rho(s,t) \equiv \bCov_{\CBP}(Z_{\CBP}(s), Z_{\CBP}(t)). $ 
The following theorem gives an analytic expression for $\bCov_{\CBP}(R_G(t_1),R_G(t_2))$ at $t_1$ and $t_2$ divisible by $L$.

\begin{theorem}\label{thm:Cov_cbp}
For $t_1=a_1L < t_2=a_2L$ where $a_1, a_2\in \{0,1,\dots,m\}$, we have
\begin{align*}
& \bCov_{\CBP}(R_G(t_1), R_G(t_2))  \\
&\hspace{10mm} = c_1 q_1(a_1,a_2) + c_2 q_2(a_1,a_2) + c_3 q_3(a_1,a_2) - c_0^2 p_1(a_1)p_1(a_2), 
\end{align*}
with $c_0$, $c_1$, $c_2$, $c_3$ provided in Theorem \ref{thm:Vcbp}, and
\begin{align*}
q_1(a_1,a_2) & = \tfrac{2a_1(m-a_2)}{m(m-1)}, \\
q_2(a_1,a_2) & = \tfrac{a_1(m-a_2)(m-2a_1+2a_2-2)}{m(m-1)(m-2)}, \\
q_3(a_1,a_2) & = \tfrac{4a_1(m-a_2)\{(a_1-1)(m-a_1-1)+(a_2-a_1)(m-a_1-2)\}}{m(m-1)(m-2)(m-3)}.
\end{align*}
\end{theorem}
The proof of this theorem is in Appendix \ref{sec:Cov_cbp_Proof}.


Then, for $t=aL$, $a\in\{1,\dots,m-1\}$, we have
\begin{align*}
C(t) & = \tfrac{m(m-1)( h_1(a,m) c_1 + h_2(a,m) c_2 + h_3(a,m) c_3)}{2La(m-a)( h_4(a,m) (2c_1 + c_2) + h_5(a,m) c_3 + h_6(a,m) c_0^2)}
\end{align*}
where
\begin{align*}
h_1(a,m) & = 2m(m-2)(m-3), \\
h_2(a,m) & = (m-3)((m-2a)^2-2m), \\
h_3(a,m) & = -4(m-2a)^2+4m, \\
h_4(a,m) & = m(m-1)(m-2)(m-3)  \\
h_5(a,m) & = 4m(m-1)(a-1)(m-a-1), \\
h_6(a,m) & = -4a(m-a)(m-2)(m-3).
\end{align*}

Based on the above results, the $p$-value approximation is reasonably well for low to moderate dimension, but not that well when the dimension is high (details see in Section \ref{sec:pval_check}, Table \ref{tab:cv}).  The reason is that when the dimension is high, the convergence of $Z_{\CBP}(t)$ to the Gaussian distribution is low when $t$ closes to the two ends and, for finite samples, $Z_{\CBP}(t)$ could be quite skewed. Figure \ref{fig:skewness} plots the skewness of $Z_{\CBP}(t)$ estimated from 100,000 random circular block permutations for two data sequences generated based on model M2 under scenarios 2 and 3 provided in Section \ref{sec:pval_check}.  We can see that the skewness is quite severe when $t$ is close to the two ends of the sequence and when the dimension is high.  \cite{chen2015graph} discussed the same issue under the permutation framework, here, we adopt a similar treatment for the circular block permutation by adding an extra term $S(t)$ to correct for the skewness.  This term varies across $t$ to address for the different extend of the skewness across $t$.
\begin{align}\label{eq:pval2}
P(\max_{n_0\leq t\leq n_1} Z_{\CBP}(t)>b) & \approx b\phi(b) \sum_{n_0\leq t\leq n_1} S(t) C(t) \nu\left(\sqrt{2b^2 C(t) } \right),
\end{align}
where
$S(t)  = \tfrac{\exp\left(\tfrac{1}{2}(b-\hat{\theta}_b(t))^2+\tfrac{1}{6}\gamma(t) \hat{\theta}^3_b(t) \right)}{\sqrt{1+\gamma(t)\hat{\theta}_b(t)}},$
with $\gamma(t) = \ep_{\CBP}(Z^3_{\CBP}(t))$ and $\hat{\theta}_b(t) = (-1+\sqrt{1+2b\gamma(t)})/\gamma(t)$.

To get an exact analytic expression for $\ep_{\CBP}(Z^3_{\CBP}(t))$, one needs to figure out all possible configurations of 3 edges, and whether any of the six end nodes are within a block or not.  Thus, even for only calculation $t=aL, a=1,2,\dots,m$, the analytic expression is very complicated and needs quite some time to run in R.  It turns out that $\gamma(t)$ can be reasonably well approximated by $\ep_\per(Z^3_\per(t))$, which can be instantly computed in R with its exact analytic expression provided in \cite{chen2015graph}.  Figure \ref{fig:skewness} plots $\ep_\per(Z^3_\per(t))$ (red line) on top of estimated $\ep_{\CBP}(Z^3_{\CBP}(t))$ (black dots), and we can see that $\ep_\per(Z^3_\per(t))$ provides a good estimate to $\ep_{\CBP}(Z^3_{\CBP}(t))$.  Hence, when we apply \eqref{eq:pval2} to approximate the $p$-value, we use $\ep_\per(Z^3_\per(t))$ to estimate $\gamma(t)$ in computing $S(t)$.

\begin{figure}[!htp]
\caption{Skewness $\ep_{\CBP}(Z^3_{\CBP}(t))$ estimated from 100,000 random circular block permutations (black dots), and $\ep_\per(Z^3_\per(t))$ computed from its exact analytic expression (red line). \label{fig:skewness}}	
\hspace{8mm} $d=100$ \hspace{48mm} $d=1,000$  \hspace{35mm} \\
\vspace{-0.2em}
	\includegraphics[width=.48\textwidth]{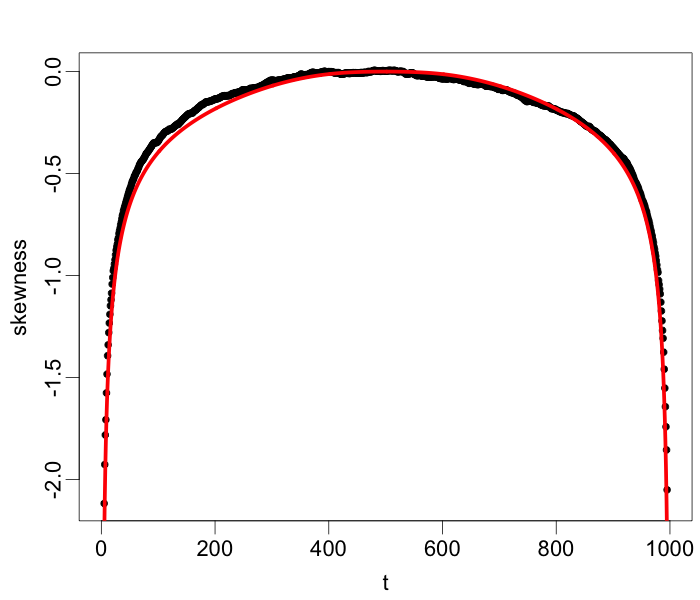} 
	\includegraphics[width=.48\textwidth]{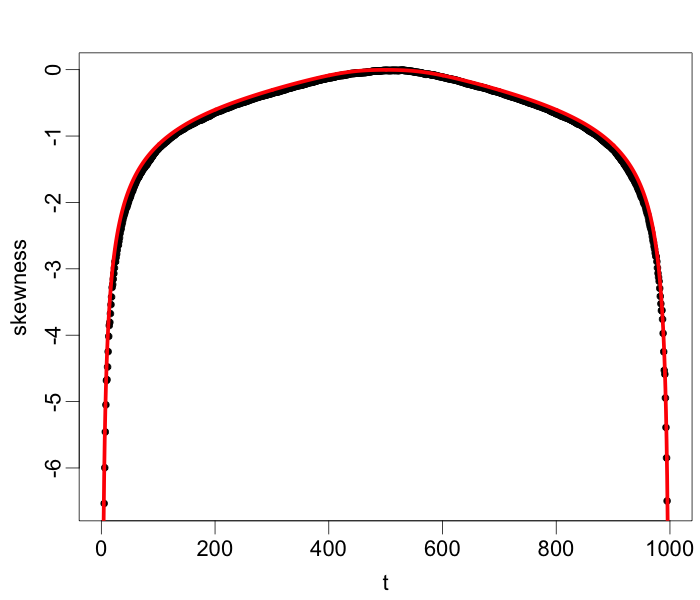} 
\end{figure}

\subsection{Check analytic $p$-value approximations} \label{sec:pval_check}
Here, we check how accurate the analytic formulas provided in \ref{sec:pval_formula} in approximating the $p$-values.  We compare the analytic $p$-value approximations obtained through asymptotic results \eqref{eq:pval1} (denoted by ``A1") and after skewness correction \eqref{eq:pval2} (denoted by ``A2") with the $p$-value estimated from 100,000 random circular block permutations (denoted by ``CBP").  We generate data from autoregressive and/or moving average models and consider the following three scenarios:
\begin{enumerate}[{Scenario} 1:]
	\item $d=10$, noises generated from the Gaussian distribution.
	\item $d=100$, noises generated from $t_5$ distribution. 
	\item $d=1,000$, noises generated from the Laplace distribution.
\end{enumerate}

We use function \texttt{arima.sim} to first generate $d$ independent sequences of time series data of length $n=1,000$: $\bz_1, \bz_2, \dots, \bz_n \in \mathbb{R}^d$. Then, let $\by_t = \Sigma^{1/2} \bz_t$, where the $(i,j)$ element of $\Sigma$ is $0.6^{|i-j|}$.  The methods are applied to the data sequence $\{\by_1,\by_2,\dots, \by_n\}$. We consider five autoregressive and/or moving average models in generating $\bz_t$'s.
\begin{itemize}
	\item M1: AR(1) with parameter 0.1.
	\item M2: AR(2) with parameters 0.1 and 0.05.
	\item M3: MA(1) with parameter 0.1.
	\item M4: MA(2) with parameters 0.1 and 0.05.
	\item M5: ARMA(1,1) with parameters 0.1 and 0.1.
\end{itemize}

\begin{table}[h]
\centering
\caption{Critical values at 0.05 significance level based on the asymptotic results (``A1"), the skewness corrected version  (``A2"), and 100,000 circular block permutations (``CBP").}\label{tab:cv}
\begin{tabular}{|c|l||c|c|c||c|c|c||c|c|c|}
\hline
& & \multicolumn{3}{|c||}{Scenario 1 ($d=10$)} & \multicolumn{3}{|c||}{Scenario 2 ($d=100$)} & \multicolumn{3}{|c|}{Scenario 3 ($d=1000$)}  \\ \cline{3-11}
& & A1 & A2 & CBP & A1 & A2 & CBP & A1 & A2 & CBP \\ \hline  \hline
 \multirow{4}{*}{M1} & $L=2$ & 3.05 & 2.94 & 2.92 & 2.99 & 2.72 & 2.71 & 2.95 & 2.38 & 2.44 \\ \cline{2-11}
    & $L=5$ & 3.05 & 2.94 & 2.92 & 2.99 & 2.72 & 2.71 & 2.95 & 2.38 & 2.40 \\ \cline{2-11} 
    & $L=10$ & 3.05 & 2.94 & 2.92 & 2.99 & 2.72 & 2.71 & 2.95 & 2.37 & 2.38 \\ \cline{2-11} 
    & $L=20$ & 3.04 & 2.94 & 2.94 & 2.99 & 2.72 & 2.71 & 2.95 & 2.37 & 2.36 \\ \hline \hline
 \multirow{4}{*}{M2} & $L=2$ & 3.05 & 2.94 & 2.92 & 3.00 & 2.77 & 2.73 & 2.96 & 2.47 & 2.52 \\ \cline{2-11}
    & $L=5$ & 3.05 & 2.94 & 2.91 & 3.00 & 2.77 & 2.74 & 2.96 & 2.47 & 2.51 \\ \cline{2-11} 
    & $L=10$ & 3.05 & 2.94 & 2.91 & 3.00 & 2.77 & 2.72 & 2.96 & 2.47 & 2.51 \\ \cline{2-11} 
    & $L=20$ & 3.05 & 2.95 & 2.88 & 3.00 & 2.77 & 2.69 & 2.96 & 2.46 & 2.50 \\ \hline \hline
 \multirow{4}{*}{M3} & $L=2$ & 3.05 & 2.95 & 2.94 & 3.00 & 2.73 & 2.74 & 2.96 & 2.40 & 2.48 \\ \cline{2-11}
    & $L=5$ & 3.05 & 2.95 & 2.95 & 3.01 & 2.74 & 2.76 & 2.96 & 2.40 & 2.48 \\ \cline{2-11} 
    & $L=10$ & 3.06 & 2.96 & 2.96 & 3.01 & 2.74 & 2.75 & 2.96 & 2.40 & 2.48 \\ \cline{2-11} 
    & $L=20$ & 3.06 & 2.96 & 2.98 & 3.01 & 2.74 & 2.77 & 2.96 & 2.40 & 2.47 \\ \hline \hline
 \multirow{4}{*}{M4} & $L=2$ & 3.05 & 2.94 & 2.91 & 3.00 & 2.71 & 2.72 & 2.97 & 2.52 & 2.55 \\ \cline{2-11}
    & $L=5$ & 3.05 & 2.94 & 2.90 & 2.99 & 2.71 & 2.69 & 2.97 & 2.52 & 2.55 \\ \cline{2-11} 
    & $L=10$ & 3.05 & 2.94 & 2.91 & 3.00 & 2.71 & 2.72 & 2.97 & 2.53 & 2.57 \\ \cline{2-11} 
    & $L=20$ & 3.05 & 2.94 & 2.89 & 3.00 & 2.71 & 2.72 & 2.98 & 2.53 & 2.59 \\ \hline \hline
 \multirow{4}{*}{M5} & $L=2$ & 3.05 & 2.96 & 2.93 & 3.01 & 2.78 & 2.76 & 2.98 & 2.59 & 2.60 \\ \cline{2-11}
    & $L=5$ & 3.05 & 2.96 & 2.93 & 3.01 & 2.78 & 2.77 & 2.98 & 2.59 & 2.58 \\ \cline{2-11} 
    & $L=10$ & 3.05 & 2.96 & 2.93 & 3.01 & 2.78 & 2.78 & 2.97 & 2.58 & 2.56 \\ \cline{2-11} 
    & $L=20$ & 3.06 & 2.96 & 2.95 & 3.01 & 2.78 & 2.80 & 2.97 & 2.58 & 2.56 \\ \hline \hline
\end{tabular}
	
\end{table}

We checked for several block sizes: $L=2, 5, 10, 20$.  The critical values based on ``A1", ``A2", and ``CBP" are provided in Table \ref{tab:cv}.  Here, results under ``CBP" can be viewed as good estimates of the true critical values even though they are time-consuming to obtain.  We compare results under ``A1" and ``A2" with those under ``CBP".  We see that, under scenario 1, when the dimension is low, both analytic $p$-value approximations work reasonably well with those under ``A2" very close the corresponding ones under ``CBP".  Under scenarios 2 and 3, the dimension is higher, we see that the analytic $p$-value approximation only based on the asymptotic results is doing poorly, while the analytic $p$-value approximation after skewness correction are still close to those obtained through 100,000 circular block permutations.  The same pattern goes for different models and different choices of $L$'s.  Based on these simulation studies, we see that the analytic $p$-value approximation after skewness correction \eqref{eq:pval2} provides reasonably accurate estimates to the $p$-value under the circular block permutation framework.

\section{Discussion}
\label{sec:discussion}

In this section, we discuss the choice of $L$ through a data driven way, and the extension of the proposed framework to accommodate multiple change-points.


\subsection{Choice of $L$}\label{sec:Lchoice}
One practical question in applying this circular block permutation framework is the choice of $L$.  Here, we provide a data driven way to choose $L$.  Figure \ref{fig:L_Z} plots the $Z_{G,\CBP}(t)$ under different choices of $L$'s.

\begin{figure}[!htp]
\centering
\caption{Plots of the scan statistic $Z_{G,\CBP}(t)$ under different choices of $L$'s.  Data generated from the multivariate autoregression model $\by_t = \rho\by_{t-1} + \bvareps_t,~t=1,\dots,n$, with  $\by_0\sim \mathcal{N}(\mathbf{0},\frac{1}{1-\rho^2}\Sigma), \bvareps_1,\dots,\bvareps_n \overset{iid}{\sim}\mathcal{N}(\mathbf{0}, \Sigma), d=100,~n=1,000$.  The $(i,j)$th element of $\Sigma$ is $0.6^{|i-j|}$.  Top panel: $\rho=0$; middle panel: $\rho=0.1$; bottom panel: $\rho=0.2$. \label{fig:L_Z}}
\includegraphics[width=.9\textwidth]{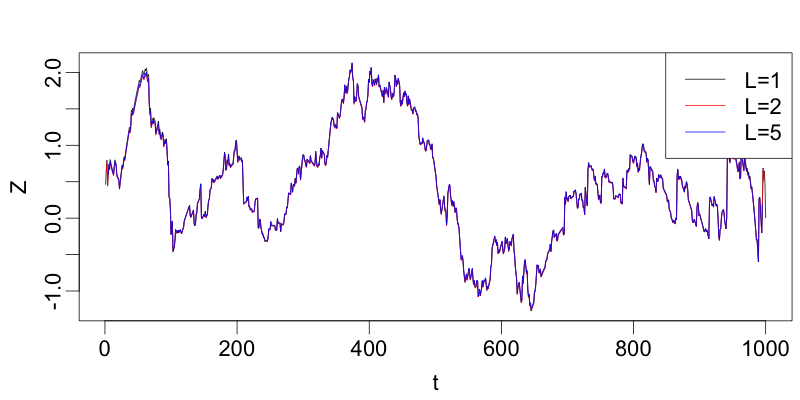}  \\
\includegraphics[width=.9\textwidth]{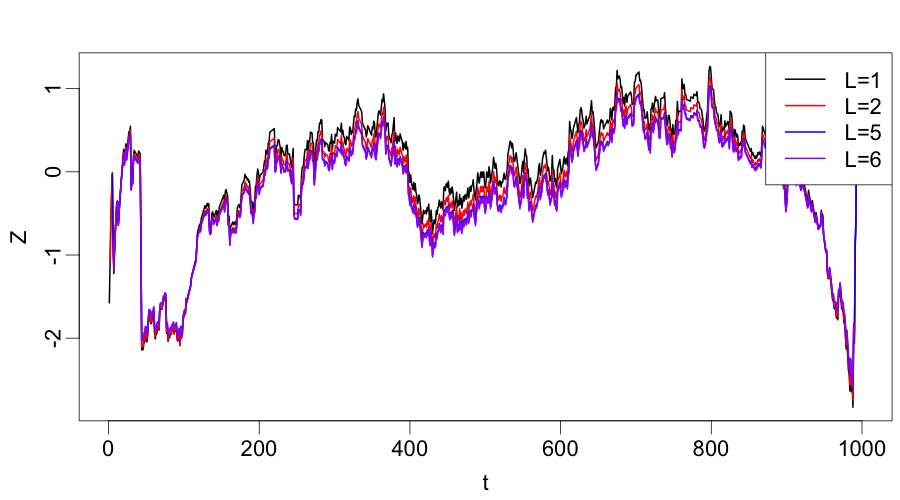} \\
\includegraphics[width=.9\textwidth]{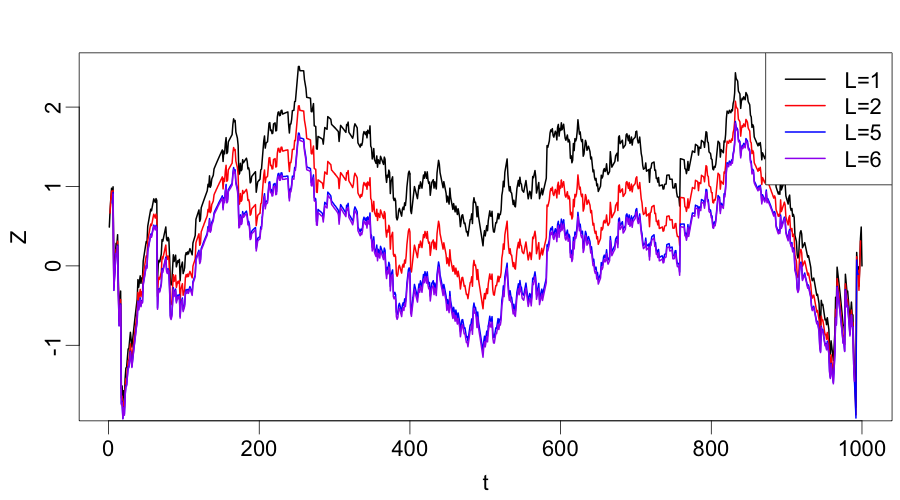}
\end{figure} 

In the top panel, the observations are independent: $\by_1, \dots, \by_n \overset{iid}{\sim} \mathcal{N}(\mathbf{0}, \Sigma)$, and the $(i,j)$th element of $\Sigma$ is $0.6^{|i-j|}$. The dimension of the data is 100 and the length of the sequence is 1,000.  
We see that the scan statistic $Z_{G,\CBP}(t)$ under $L=1$ (permutation), $L=2$, $L=5$ are almost the same.  On the other hand, in the middle and bottom panels, the data sequence is from a multivariate autoregression model: $\by_t = \rho\by_{t-1} + \bvareps_t,~t=1,\dots,n$, with  $\by_0\sim \mathcal{N}(\mathbf{0},\frac{1}{1-\rho^2}\Sigma), \bvareps_1,\dots,\bvareps_n \overset{iid}{\sim}\mathcal{N}(\mathbf{0}, \Sigma)$.  Here, $\rho=0.1$ for the middle panel, and $\rho=0.2$ for the bottom panel.  First of all, when the observations are not independent, the scan statistic under $L=2$ no longer overlaps with that under permutation ($L=1$).  The stronger the autocorrelation is, the larger the discrepancy between the two curves under $L=1$ and $L=2$.  Secondly, as $L$ increases, the curves becomes more similar.  For example, in the middle panel, the two curves under $L=5$ and $L=6$ almost overlap with each other; in the bottom panel, the curve under $L=6$ is also close to that under $L=5$, showing that $L$ around this range is close to enough to take care of the local dependence in the sequence.  Hence, we could choose $L$ to be the value that the scan statistic no longer changes or changes in a negligible amount.  

\begin{figure}[!htp]
\centering
\caption{Plots of the maximum scan statistic $\max_{n_0\leq t\leq n_1} Z_{G,\CBP}(t)$ under different choices of $L$'s based on the same data sets in Figure \ref{fig:L_Z}.  Top panel: $\rho=0$; middle panel: $\rho=0.1$; bottom panel: $\rho=0.2$. In the right panel, the ratio of $\max_{n_0\leq t\leq n_1} Z_{G,\CBP}(t)$ under $L+1$ over that under $L$ is plotted.  The horizontal line is at 0.99. \label{fig:L_Zm}}
\includegraphics[width=.47\textwidth]{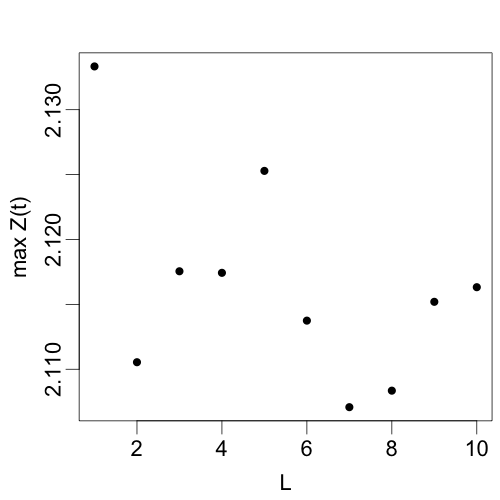}	
\includegraphics[width=.47\textwidth]{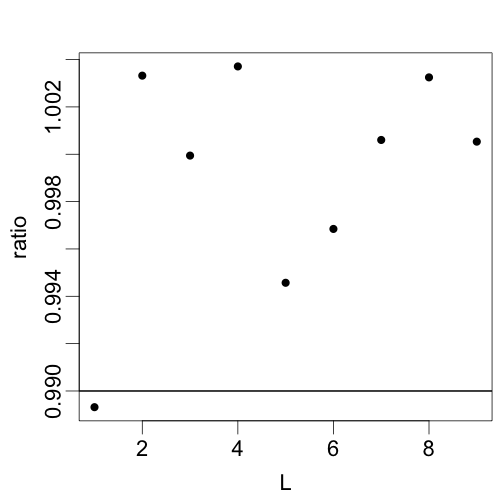}	
\includegraphics[width=.47\textwidth]{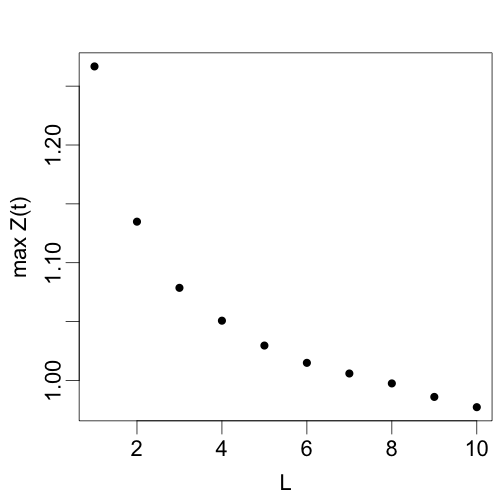}	
\includegraphics[width=.47\textwidth]{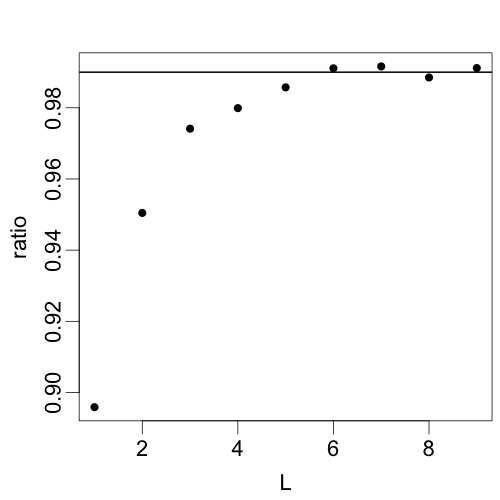}	
\includegraphics[width=.47\textwidth]{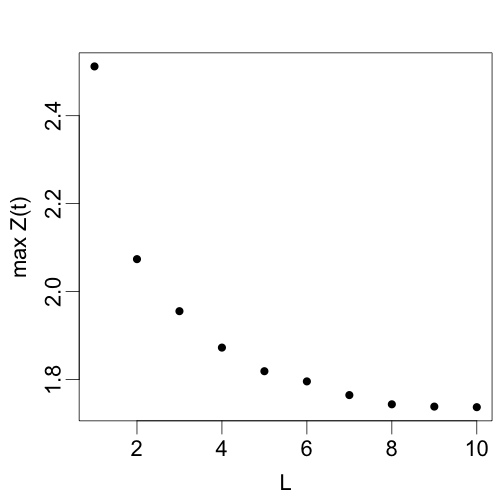}	
\includegraphics[width=.47\textwidth]{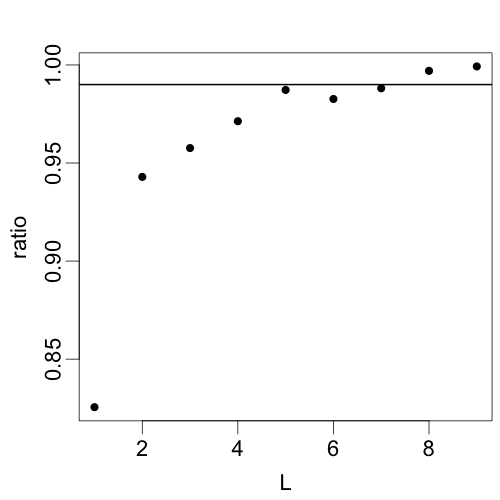}	
\end{figure} 

To be more specific, we plot the maximum scan statistic $\max_{n_0\leq t\leq n_1} Z_{G,\CBP}(t)$ over $L$ for the three data sequences and they are shown in Figure \ref{fig:L_Zm} (left panels).  We see that the maximum scan statistic scatter around under the independent case (top panel), and the maximum scan statistic is decreasing for $L$ from 1 to 10 under the autocorrelated cases (middle panel and bottom panel).  In the right panel, the ratio of the maximum scan statistic under $L+1$ over that under $L$ is plotted.  In each ratio plot, a horizontal line at 0.99 is added, which appears to be a reasonable threshold to use in choosing $L$.  We could set $L$ to be the smallest value such that the ratio goes above 0.99.  Under this criterion, we could set $L=2$, $L=6$, and $L=8$ for the three data sequences, respectively. 

\subsection{Extension to multiple change-points} \label{sec:multipleCP}

We focused on the single change-point alternative so far.  The proposed framework could be extended to the changed interval alternative up to some modifications: For any pair of times $t_1<t_2$, we could construct the test statistic $Z_{G,\CBP}(t_1,t_2)$ to test $\{\by_{t_1},\dots,\by_{t_2-1}\}$ against $\{\by_{t_2},\dots,\by_n,\by_1,\dots,\by_{t_1-1}\}$ in a similar way as $Z_{G,\CBP}(t)$, which tests $\{\by_1,\dots,\by_t\}$ against $\{\by_{t+1},\dots,\by_n\}$.  The scan statistic for the changed-interval alternative can then be defined as  
$$\max_{1\leq t_1<t_2\leq n,\ l_0\leq t_2-t_1\leq l_1} Z_{G,\CBP}(t_1,t_2), \quad (l_0, l_1 \text{ prespecified}).$$
Other treatments could follow accordingly.  

For more complicated scenarios with possibly multiple change-points, we could apply the scan statistic  for the single change-point alternative recursively in a binary segmentation procedure \citep{vostrikova1981detecting}, or apply the scan statistic for the changed-interval alternative recursively in a circular binary segmentation procedure \citep{olshen2004circular}. 

%

\section*{Acknowledgments}
Hao Chen is supported in part by NSF award DMS-1513653.

\bibliographystyle{imsart-nameyear}
\bibliography{cbp}

\appendix

\section{Proofs for theorems}

\subsection{Proof of Theorem \ref{thm:Vcbp}}
\label{sec:Vcbp_Proof}
We have
\begin{align*}
& \bE_{\CBP}(R_G^2(t))\\
& = \sum_{(i,j), (u,v)\in G} P(g_{\pi_{\CBP}(i)}(t)\neq g_{\pi_{\CBP}(j)}(t),\ g_{\pi_{\CBP}(u)}(t) \neq \pi_{\pi_{\CBP}(v)}(t)) \\
& = \sum_{(i,j)\in G} P(g_{\pi_{\CBP}(i)}(t)\neq g_{\pi_{\CBP}(j)}(t)) \\
& \ + \sum_{(i,j), (i,u)\in G;\ j\neq u} P(g_{\pi_{\CBP}(i)}(t)\neq g_{\pi_{\CBP}(j)}(t),\ g_{\pi_{\CBP}(i)}(t) \neq \pi_{\pi_{\CBP}(u)}(t)) \\
& \ + \sum_{\footnotesize \tiny \begin{array}{c}(i,j), (u,v)\in G\\ i,j,u,v\text{ all different}\end{array}} P(g_{\pi_{\CBP}(i)}(t)\neq g_{\pi_{\CBP}(j)}(t),\ g_{\pi_{\CBP}(u)}(t) \neq \pi_{\pi_{\CBP}(v)}(t)).
\end{align*}
The first part of the summation is $\bE_{\CBP}(R_G(t))$.  In the following, we figure out the second and third part of the summation.

\subsubsection{$P(g_{\pi_{\CBP}(i)}(t)\neq g_{\pi_{\CBP}(j)}(t),\ g_{\pi_{\CBP}(i)}(t)\neq \pi_{\pi_{\CBP}(u)}(t)), \ j\neq u$}

When $\delta_{ij}, \delta_{iu}, \delta_{ju} \geq L$, $i,j,u$ are all in different blocks, and the probability is $p_2(a)$.  When at least one of $\delta_{ij}, \delta_{iu}, \delta_{ju} \geq L$ is less than $L$, we need to consider scenarios that some of the indices could be in the same block.  In the following, we consider when the event $\{g_{\pi_{\CBP}(i)}(t)\neq g_{\pi_{\CBP}(j)}(t),\ g_{\pi_{\CBP}(i)}(t)\neq \pi_{\pi_{\CBP}(u)}(t)\}$ could happen.

When $\delta_{ij} <L$, $\delta_{iu}, \delta_{ju} \geq L$, since $\pi_{\CBP}(i)$ and $\pi_{\CBP}(j)$ need to be in different blocks, $i$ and $j$ need to be in different blocks. And the probability is
\begin{align*}
\tfrac{\delta_{ij}}{L} p_2(a).
\end{align*}
Similarly, when $\delta_{iu} <L$, $\delta_{ij}, \delta_{ju} \geq L$, the probability is
\begin{align*}
\tfrac{\delta_{iu}}{L} p_2(a).
\end{align*}

When $\delta_{ju} <L$, $\delta_{ij}, \delta_{iu} \geq L$, since $\pi_{\CBP}(j)$ and $\pi_{\CBP}(u)$ can be in the same block, so the probability is
\begin{align*}
\tfrac{l-\delta_{ju}}{L} p_1(a) +\tfrac{\delta_{ju}}{L} p_2(a).
\end{align*}

When $\delta_{ij}, \delta_{iu} <L$, $\delta_{ju} \geq L$, then $\delta_{ju} = \delta_{ij}+\delta_{iu}$.   $j$ and $u$ will always be in different blocks.  We need $i$ be in a block
different from $j$'s and $u$'s block so that the probability is positive.  So the probability is
\begin{align*}
\tfrac{\delta_{ju}-L}{L} p_2(a).
\end{align*}

When $\delta_{ij}, \delta_{ju} <L$, $\delta_{iu} \geq L$, then $\delta_{iu} = \delta_{ij}+\delta_{ju}$.  $i$ and $j$ needs to be in different blocks, so the probability is
\begin{align*}
\tfrac{\delta_{ij}-(\delta_{iu}-L)}{L} p_1(a)+\tfrac{\delta_{iu}-L}{L} p_2(a) = \tfrac{L-\delta_{ju}}{L} p_1(a) + \tfrac{\delta_{iu}-L}{L} p_2(a).
\end{align*}
Similarly for $\delta_{iu}, \delta_{ju} <L$.

When $\delta_{ij}, \delta_{iu}, \delta_{ju} <L$, $i$ needs to be in different block from $j$ and $u$. So the probability is
\begin{align*}
 \tfrac{\min(\delta_{ij},\delta_{iu})}{L} I(\max(\delta_{ij},\delta_{iu},\delta_{ju})\neq \delta_{ju}) p_1(a).
\end{align*}

\subsubsection{$P(g_{\pi_{\CBP}(i)}(t)\neq g_{\pi_{\CBP}(j)}(t),\ g_{\pi_{\CBP}(u)}(t) \neq \pi_{\pi_{\CBP}(v)}(t))$, $i,j,u,v$ all different}

When $i,j,u,v$ are all different, there are $\binom{4}{2}=6$ index-pairs among them.
If all pairwise index distances are greater than or equal to $L$, this probability is $p_3(a)$.  In the following, we consider scenarios that at least one of the six index distances is less than $L$.

\begin{enumerate}[1)]
\item One distance $<L$.

If $\delta_{ij}<L$ and all other pairwise distances are $\geq L$, then this probability is
$$\tfrac{\delta_{ij}}{L} p_3(a).$$
Similarly, if only $\delta_{uv}<L$, the probability is 
$$\tfrac{\delta_{uv}}{L} p_3(a).$$

If only $\delta_{iu}<L$, the probability is
$$\tfrac{L-\delta_{iu}}{L} p_2(a) + \tfrac{\delta_{iu}}{L} p_3(a).$$
Similar for only $\delta_{iv}<L$, or $\delta_{ju}<L$, or $\delta_{jv}<L$.

\item Two distances $<L$.

If only $\delta_{ij},\delta_{uv}<L$, 
  then the probability is
$$\tfrac{(\min(\delta_{ij}, b_{ij,uv}+\delta_{uv})-b_{ij,uv})_+ + (\min(\delta_{ij},b_{ij,uv}+\delta_{uv}-l))_+}{L} p_3(a).$$

If only $\delta_{ij}, \delta_{iu}<L$, we must have $\delta_{ui}+\delta{ij} = \delta_{uj}$, the probability is 
$$\tfrac{L-\delta_{ui}}{L} p_2(a) + \tfrac{\delta_{ij}+\delta_{iu}-L}{L} p_3(a).$$
Similar for other 7 similar cases.

If only $\delta_{iu}, \delta_{iv} <L$, then we have $\delta_{uv} = \delta_{ui} + \delta_{iv}$, the probability is
$$\tfrac{2L-\delta_{uv}}{L} p_2(a) + \tfrac{\delta_{uv}-L}{L} p_3(a).$$
Similar for other 3 similar cases.

If only $\delta_{iu},\delta_{jv}<L$, the probability is
$$\tfrac{l-\delta_{iu}-\delta_{jv}+x(iu,jv)}{L} p_1(a) + \tfrac{\delta_{iu}+\delta_{jv}-2x(iu,jv)}{L} p_2(a) + \tfrac{x(iu,jv)}{L} p_3(a).$$
Similar for the case that only $\delta_{iv},\delta_{ju}<L$.

\item Three distances $<L$.

If only $\delta_{ij}, \delta_{uv}, \delta_{iu} <L$, then the order of the four indices must be $(j,i,u,v)$ or the reverse, and $\delta_{jv} = \delta_{ji}+\delta_{iu}+\delta_{uv}$.  The probability is
$$\tfrac{L-\delta_{iu}}{L} p_2(a) + \tfrac{(\delta_{jv}-2L)_+}{L} p_3(a).$$
Similar for other 3 similar cases: replace $\delta_{iu}$ by one of $(\delta_{iv}, \delta_{ju}, \delta_{jv})$.

If only $\delta_{ij},\delta_{iu},\delta_{jv}<L$, then the order of the four indices must be $(u,i,j,v)$ or the reverse, and $\delta_{uv} = \delta_{ui}+\delta_{ij}+\delta_{jv}$.  The probability is
$$\tfrac{2L-\delta_{uv}+(\delta_{uv}-2L)_+}{L} p_1(a) + \tfrac{\delta_{uv}+\delta_{ij}-2L-2(\delta_{uv}-2L)_+}{L} p_2(a) + \tfrac{(\delta_{uv}-2L)_+}{L} p_3(a).$$
Similar for other 3 similar cases: only $\delta_{ij}, \delta_{iv}, \delta_{ju}<L$; only $\delta_{uv}, \delta_{ui}, \delta_{vj}<L$; only $\delta_{uv}, \delta_{uj}, \delta_{vi}<L$.

If only $\delta_{iu}, \delta_{iv}, \delta_{ju}<L$, then the order of the four indices must be $(j,u,i,v)$ or the reverse, and $\delta_{jv} = \delta_{ju}+\delta_{ui}+\delta_{iv}$.  The probability is
$$\tfrac{2L-\delta_{jv}+(\delta_{jv}-2L)_+}{L} p_1(a) + \tfrac{\delta_{jv}-L-2(\delta_{jv}-2L)_+}{L}p_2(a) + \tfrac{(\delta_{jv}-2L)_+}{L} p_3(a).$$
Similar for other 3 similar cases: choose 3 out of $(\delta_{iu},\delta_{iv},\delta_{ju},\delta_{jv})$.

If only $\delta_{ij}, \delta_{iu}, \delta_{ju}<L$, the probability is
$$\tfrac{\delta_{ij}}{L} p_2(a).$$
Similar for other 3 similar cases: only $\delta_{ij}, \delta_{iv}, \delta_{jv}<L$; only $\delta_{uv}, \delta_{iu}, \delta_{iv}<L$; only $\delta_{uv}, \delta_{ju}, \delta_{jv}<L$.

It's not possible for 4 cases with only three distances smaller than $l$ that share one index.  For example, only $\delta_{ij}, \delta_{iu}, \delta_{iv}<L$.

\item Four distances $<L$.

For scenarios that 4 distances $<L$ and 2 distances $\geq L$, it is impossible for 3 cases: only $\delta_{ij}, \delta_{uv}\geq L$; only $\delta_{iu},\delta_{jv}\geq L$; only $\delta_{iv},\delta_{ju}\geq L$.

If only $\delta_{ij}, \delta_{iu}\geq L$, it can be $(i,v,u,j)$ or $(i,v,j,u)$ or their reverses.  For both orders, the probability is
$$\tfrac{L-\delta_{iv}}{L} p_1(a) + \tfrac{\delta_{iu}-L}{L} p_2(a).$$
Similar for 7 other similar cases.

If only $\delta_{iu},\delta_{iv}\geq L$, it can be $(i,j,u,v)$ or $(i,j,v,u)$ or their reverses.  The probability is
$$\tfrac{\delta_{uv}}{L} p_2(a).$$
Similar for 3 other similar cases.

\item Five distances $<L$.

If only $\delta_{ij}\geq L$, the probability is
$$\tfrac{\delta_{uv}}{L} p_1(a).$$
Similar for if only $\delta_{uv}\geq L$.

If only $\delta_{iu}\geq L$, it can be $(i,j,v,u)$ or $(i,v,j,u)$ or their reverses.  The probability is
$$\tfrac{\delta_{iu}-L}{L} p_2(a) + \tfrac{\delta_{vj}}{L}I(\delta_{ij}=\delta_{iv}+\delta_{vj}) p_1(a).$$
Similar for other 3 similar cases.

\item Six distances $<L$.

If all 6 distances are $<L$.  

If $\delta_{ij}$ is the maximum distance, the probability is 
$$\tfrac{\delta_{uv}}{L} p_1(a).$$
Similar for the case that $\delta_{uv}$ is the maximum distance.

If $\delta_{iu}$ is the maximum distance, the probability is
$$\tfrac{\delta_{jv}}{L}I(\delta_{ij}=\delta_{iv}+\delta_{jv}) p_1(a).$$
Similar for 3 other similar cases.

\end{enumerate}

Summing all possible scenarios, we get the result stated in Lemma \ref{thm:Vcbp}.

\subsection{Proof of Theorem \ref{thm:asym}}\label{sec:asym_Proof}
 
To prove $\{Z_{G,\CBP}([mw]L): \epsilon \leq w \leq 1-\epsilon \}$ converges to a Gaussian process in finite dimensional distributions, we only need to show that $$(Z_{G,\CBP}([mw_1]L), Z_{G,\CBP}([mw_2]L),\dots,Z_{G,\CBP}([mw_K]L))$$ converges to multivariate Gaussian as $n\rightarrow\infty$ for any $\epsilon \leq w_1<w_2<\dots<w_K\leq 1-\epsilon $ for any fixed $K$.  For simplicity, let $t_k = [mw_k]L, k=1,\dots, K$.

To prove the above results, we take one step back.  For circular block permutation with a random starting point, in the last step of the recipe, the action is to permute the $m=n/L$ blocks.  Let $\pi(i)$ be the observed time of block $i$ after this block permutation, then $(\pi(1),\dots,\pi(m))$ is a permutation of $1,\dots,m$.  We can do the last step in the following two-step approach: (1) For each $i$, $\tilde{\pi}(i)$ is sampled uniformly from 1 to $m$; (2) only those that each value in $\{1,\dots,m\}$ is sampled once are retained.  It is easy to see that each block permutation has the same occurrence probability after these two steps.

We call the distribution resulting from only performing the first step the circular block bootstrap with a random starting point, short as {\CBB}, and use $\bP_{\CBB}, \bE_{\CBB}, \bV_{\CBB}$ to denote the probability, expectation, and variance, respectively.

Let 
\begin{align*}
Z_{G,\CBB}(t) & = -\tfrac{R_G(t) -\bE_{\CBB}(R_G(t))}{\sqrt{\bV_{\CBB}(R_G(t))}}, \\
X_{\CBB}(t) & = \tfrac{n_{\CBB}(t)-(t/L)}{\sqrt{(1-t/n)t/L}}, \text{ where } n_{\CBB}(t) = \sum_{i=1}^m I_{\tilde{\pi}(i)\leq t/L}.
\end{align*}

Then following the similar arguments for obtaining Theorems \ref{thm:Ecbp} and \ref{thm:Vcbp}, we have that, for $t=aL, a\in\{1,\dots,m\}$,
\begin{align*}
\bE_{\CBB}(R_G(t)) & = c_0 p_{1,\CBB}(a), \\
\bV_{\CBB}(R_G(t)) & = c_1 p_{1,\CBB}(a) + c_2 p_{2,\CBB}(a) + c_3 p_{3,\CBB}(a) - c_0^2 p_{1,\CBB}^2(a) \\
& :=(\sigma_{\CBB}(t))^2,
\end{align*}
where
$p_{1,\CBB}(a)  = \tfrac{2a(m-a)}{m^2}$,
$p_{2,\CBB}(a)  = \tfrac{a(m-a)}{m^2}$, 
$p_{3,\CBB}(a)  = \tfrac{4a^2(m-a)^2}{m^4}.$

We next prove the following two lemmas.
\begin{lemma}\label{thm:asym_cbb}
Under Condition 1, as $n\rightarrow\infty$, 
\begin{equation}\label{eq:ZX}
 (Z_{G,\CBB}(t_1), \dots, Z_{G,\CBB}(t_K), X_{\CBB}(t_1), \dots,X_{\CBB}(t_K)) 
\end{equation}
converges to a multivariate Gaussian distribution under CBB and the covariance matrix of 
$(X_{\CBB}(t_1), X_{\CBB}(t_2),\dots,X_{\CBB}(t_K))$
is positive definite.
\end{lemma}

\begin{lemma}\label{thm:cbbp}
Under Condition 2 or Condition 3, we have for $k=1,\dots,K$,
\begin{enumerate}
\item $\tfrac{\bV_{\CBB}(R_G(t_k))}{\bV_{\CBP}(R_G(t_k))}\rightarrow r([mw_k])$, with $r(a)$ a constant only depends on $a$.
\item $\tfrac{\bE_{\CBB}(R_G(t_k))-E_{\CBP}(R_G(t_k))}{\sqrt{\bV_{\CBB}(R_G(t_k))}}\rightarrow 0.$
\end{enumerate}
\end{lemma}

From Lemma \ref{thm:asym_cbb},  $(Z_{G,\CBB}(t_1), Z_{G,\CBB}(t_2),\dots, Z_{G,\CBB}(t_K))$ conditioning on $(X_{\CBB}(t_1), X_{\CBB}(t_2),\dots,X_{\CBB}(t_K))$ converges to multivariate normal under \CBB.  Since  $(Z_{G,\CBB}(t_1), Z_{G,\CBB}(t_2),\dots, Z_{G,\CBB}(t_K))| X_{\CBB}(t_1)=0,$ $X_{\CBB}(t_2)=0, \dots, X_{\CBB}(t_K)=0)$ under CBB has the same distribution as $(Z_{G,\CBB}(t_1), Z_{G,\CBB}(t_2),\dots, Z_{G,\CBB}(t_K))$ under \CBP, and notice that
$$Z_{G,\CBP}(t) = \tfrac{\bV_{\CBB}(R_G(t))}{\bV_{\CBP}(R_G(t))}\left(Z_{G,\CBB}(t) -   \tfrac{\bE_{\CBB}(R_G(t))-E_{\CBP}(R_G(t))}{\sqrt{\bV_{\CBB}(R_G(t))}}\right).$$
Given Lemma \ref{thm:cbbp}, we conclude that $(Z_{G,\CBP}(t_1), Z_{G,\CBP}(t_2),\dots,Z_{G,\CBP}(t_K))$ converges to a multivariate Gaussian distribution under \CBP.  Next, we prove the two lemmas.

\subsubsection{Proof for Lemma \ref{thm:asym_cbb}}  To show that \eqref{eq:ZX} converges to a multivariate Gaussian distribution, we only need to show that $\sum_{k=1}^K(c_{1k} Z_{G,\CBB}(t_k) + c_{2k} X_{\CBB}(t_k))$ converges to a normal distribution for any fixed $\{c_{1k}\}$ and $\{c_{2k}\}$.
If $\bV_{G,\CBB}(\sum_{k=1}^K(c_{1k} Z_{G,\CBB}(t_k) + c_{2k} X_{\CBB}(t_k))=0$, $\sum_{k=1}^K(c_{1k} Z_{G,\CBB}(t_k) + c_{2k} X_{\CBB}(t_k)$ is degenerating.  For non-degenerating case, let $$\sigma_0^2 = \bV_{G,\CBB}(\sum_{k=1}^K(c_{1k} Z_{G,\CBB}(t_k) + c_{2k} X_{\CBB}(t_k)).$$  

We prove the Gaussianity of $\sum_{k=1}^K(c_{1k} Z_{G,\CBB}(t_k) + c_{2k} X_{\CBB}(t_k))$ by the Stein's method.
Consider sums of the form $W=\sum_{i\in{\cal J}} \xi_i,$
where $\mathcal{J}$ is an index set and $\xi$ are random variables with $\bE \xi_i=0$, and $\bE (W^2)=1$.  The following assumption restricts the dependence between $\{\xi_i:~i \in \mathcal{J}\}$.
\begin{assumption} \cite[p.\, 17]{chen2005stein}
  \label{assump:LD} 
For each $i\in{\cal J}$ there exists $K_i \subset L_i \subset {\cal J}$ such that $\xi_i$ is independent of $\xi_{K_i^c}$ and $\xi_{K_i}$ is independent of $\xi_{L_i^c}$.
\end{assumption}
We will use the following existing theorem in proving Theorem \ref{thm:asym}.
\begin{theorem}\label{thm:3.4} \cite[Theorem 3.4]{chen2005stein}
Under Assumption \ref{assump:LD}, we have
$$\sup_{h\in Lip(1)} |E(h(W)) - E(h(Z))| \leq \delta$$
where $Lip(1) = \{h: \mathbb{R}\rightarrow \mathbb{R}; \|h^\prime\|\leq 1 \}$, $Z$ has ${\cal N}(0,1)$ distribution and
 $$\delta = 2 \sum_{i\in{\cal J}} (E|\xi_i \eta_i\theta_i| + |E(\xi_i\eta_i)|E|\theta_i|) + \sum_{i\in{\cal J}} E|\xi_i\eta_i^2|$$
with $\eta_i = \sum_{j\in K_i}\xi_j$ and $\theta_i = \sum_{j\in L_i} \xi_j$, where $K_i$ and $L_i$ are defined in Assumption \ref{assump:LD}.
\end{theorem}

We adopt the same notations with the index set $\mathcal{J}=\{e\in G\}\cup\{1,\dots,n\}$.  Denote the end nodes of an edge $e$ to be $e_-$ and $e_+$. Let 
$$\xi_{e,k} = \tfrac{I_{g_{\tilde{\pi}(e_-)}(t_k)\neq g_{\tilde{\pi}(e_+)}(t_k)} - p_{1,\CBB}(t_k)}{\sigma_{\CBB}(t_k)}.$$
Since $I_{g_{\tilde{\pi}(e_-)}(t_k)\neq g_{\tilde{\pi}(e_+)}(t_k)}\in \{0,1\}$ and $p_{1,\CBB}(t_k)\in[0,1]$, we have that $$|\xi_{e,k}|\leq \tfrac{1}{\sigma_{\CBB}(t_k)}.$$

Let $$\xi_{i,k} = \tfrac{I_{\tilde{\pi}(i)\leq t_k/L}-w_k}{\sqrt{m w_k(1-w_k)}}.$$
Similarly, we have $$|\xi_{i,k}|\leq \tfrac{1}{\sqrt{m w_k(1-w_k)}}.$$

Let $\xi_e = \sum_{k=1}^K c_{1k} \xi_{e,k}/\sigma_0$, $\xi_i = \sum_{k=1}^K c_{2k} \xi_{i,k}/\sigma_0$, and $W=\sum_{j\in \mathcal{J}} \xi_j = \sum_{k=1}^K (c_{1k} Z_{G,\CBB}(t_k) + c_{2k} X_{\CBB}(t_k))/\sigma_0$.  Then $\bE_{\CBB}(W)=0, \bE_{\CBB}(W^2) = 1$.  

Let $c_M = \max(\sum_{k=1}^K |c_{1k}|, \sum_{k=1}^K |c_{2k}|)$, $\sigma_1 = \sigma_0\times \min_k \sigma_{\CBB}(t_k)$, $\sigma_2 = \sigma_0 \times \min_k \sqrt{m w_k(1-w_k)}$, then
$$|\xi_e|\leq \tfrac{c_M}{\sigma_1},\ \forall e\in G; \quad |\xi_i|\leq \tfrac{c_M}{\sigma_2},\ \forall  i\in\{1,\dots,n\}.$$
Based on Remark \ref{remark:coef}, $c_1+c_2+c_3\geq c_0^2$, we have, with $a_k=[mw_k]$, 
\begin{align*}
\bV_\CBB(R_G(t_k))& \geq  c_1(p_{1,\CBB}(a_k)-p_{1,\CBB}^2(a_k)) + c_2(p_{2,\CBB}(a_k)-p_{1,\CBB}^2(a_k)) \\
& = c_1\tfrac{a_k(m-a_k)}{m(m-1)} + (c_1+c_2)\tfrac{a_k(m-a_k)}{m(m-1)} \left(1-\tfrac{4a_k(m-a_k)}{m(m-1)}\right) \\
& \geq c_1\tfrac{a_k(m-a_k)}{m(m-1)} - (c_1+c_2)\tfrac{a_k(m-a_k)}{m(m-1)^2}.
\end{align*}
It can be shown that $c_2=o(mc_1)$ (detailed arguments seen in the proof for Lemma \ref{thm:cbbp}).  So $\bV_\CBB(R_G(t_k))$ is at least of order $c_1$, which is at least of order $|G|$.  It is clear that $\sigma_2 = O(n)$.

For any $A$ a set of edge(s), let $V(A)$ to be the set that contains all nodes being connected by at least one edge in $A$. Notice that $|V(A)|\leq 2|A|$ as the worst scenario occurs when all edge(s) in $A$ are disconnected.

With $A_{e,L,1}$ and $A_{e,L,2}$ defined in Section \ref{sec:asymp}, for $e\in G$, let 
\begin{align*}
S_e & = A_{e,L,1} \cup V(A_{e,L,1}),\\
T_e & = A_{e,L,2} \cup V(A_{e,L,2}), 
\end{align*}
Then $S_e$ and $T_e$ satisfy Assumption \ref{assump:LD}.

With $A_{i,L,1}$ and $A_{i,L,2}$ defined in Section \ref{sec:asymp}, for  $i = 1,\dots,n$, let
\begin{align*}
S_i & = A_{i,L,1} \cup V(A_{i,L,1}), \\
T_i & = A_{i,L,2} \cup V(A_{i,L,2}), 
\end{align*}
Then $S_i$ and $T_i$ satisfy Assumption \ref{assump:LD}.


By Theorem \ref{thm:3.4}, we have $\sum_{h\in Lip(1)} |\bE h(W) - \bE h(Z)| \leq \delta$ for $Z\sim \mathcal{N}(0,1)$, where
\allowdisplaybreaks
\begin{align*}
\delta & =  2 \sum_{i\in{\cal J}} (E|\xi_i \eta_i\theta_i| + |E(\xi_i\eta_i)|E|\theta_i|) + \sum_{i\in{\cal J}} E|\xi_i\eta_i^2| \\
& =  2 \sum_{e\in G} (E|\xi_e \eta_e\theta_e| + |E(\xi_e\eta_e)|E|\theta_e|) + \sum_{e\in G} E|\xi_e\eta_e^2| \\
& \quad +  2 \sum_{i=1}^n (E|\xi_i \eta_i\theta_i| + |E(\xi_i\eta_i)|E|\theta_i|) + \sum_{i=1}^n E|\xi_i\eta_i^2| \\
& \leq \tfrac{5c_M}{\sigma_1}\left(|A_{e,L,1}| \tfrac{c_M}{\sigma_1} + |V(A_{e,L,1})|\tfrac{c_M}{\sigma_2} \right) \left(|A_{e,L,2}| \tfrac{c_M}{\sigma_1} + |V(A_{e,L,2})|\tfrac{c_M}{\sigma_2} \right) \\
& \quad + \tfrac{5c_M}{\sigma_2}\left(|A_{i,L,1}| \tfrac{c_M}{\sigma_1} + |V(A_{i,L,1})|\tfrac{c_M}{\sigma_2} \right) \left(|A_{i,L,2}| \tfrac{c_M}{\sigma_1} + |V(A_{i,L,2})|\tfrac{c_M}{\sigma_2} \right) \\
& \leq \tfrac{5c_M}{\sigma_1}\left(\tfrac{c_M}{\sigma_1} + \tfrac{2c_M}{\sigma_2}  \right)^2|A_{e,L,1}||A_{e,L,2}| + \tfrac{5c_M}{\sigma_2}\left(\tfrac{c_M}{\sigma_1} + \tfrac{2c_M}{\sigma_2}  \right)^2|A_{i,L,1}||A_{i,L,2}|
\end{align*}
When $\sum_{e\in G} |A_{e,L,1}||A_{e,L,2}| = o(n|G|^{0.5})$ and $\sum_{i=1}^n |A_{i,L,1}||A_{i,L,2}| = o(n^{1.5})$, we have $\delta\rightarrow 0$ as $n\rightarrow \infty$.


Let $\Sigma_X$ be the covariance matrix of $(X_{\CBB}(t_1), X_{\CBB}(t_2),\dots,X_{\CBB}(t_K))$.  Follow similar arguments in \cite{chen2015graph}, $|\Sigma_X| = \tfrac{\prod_{k=1}^K (1-t_k/t_{k+1})}{\prod_{k=1}^K(1-t_k/n)}$.  So     $\Sigma_X$ is positive definite. 

\subsubsection{Proof for Lemma \ref{thm:cbbp}}
For $t=aL$, $a=[mw_k]$, $k=1,\dots,K$, we have
\begin{align*}
\bV_{\CBB}(R_G(t)) & = c_1 p_{1,\CBB}(a) + c_2 p_{2,\CBB}(a) + c_3 p_{3,\CBB}(a) - c_0^2 p_{1,\CBB}^2(a) \\
& = c_1\tfrac{4a^2(m-a)^2}{m^2(m-1)^2} + (2c_1+c_2)\tfrac{-4a^2(m-a)^2+m(m-1)a(m-a)}{m^2(m-1)^2} \\
& \quad + (c_1+c_2+c_3-c_0^2)\tfrac{4a^2(m-a)^2}{m^2(m-1)^2}, \\
\bV_{\CBP}(R_G(t)) & = c_1 p_{1}(a) + c_2 p_{2}(a) + c_3 p_{3}(a) - c_0^2 p_{1}^2(a) \\
& = c_1\tfrac{a(m-a)(a-1)(m-a-1)}{m(m-1)(m-2)(m-3)} \\
& \quad + \left(2c_1+c_2-\tfrac{4c_0^2}{m} + \tfrac{2c_0^2}{m(m-1)}\right)\tfrac{-4a^2(m-a)^2+m(m-1)a(m-a)}{m(m-1)(m-2)(m-3)} \\
& \quad + (c_1+c_2+c_3-c_0^2)\tfrac{4a^2(m-a)^2}{m^2(m-1)^2}.
\end{align*}


From Remark \ref{remark:coef}, we know that $c_1+c_2+c_3-c_0^2\geq 0$.  Using the same notations in Remark \ref{remark:coef}, under a certain blocking $\omega$, let $|G_{b_i}(\omega)|$ be the number of edges in $G$ that connect a node in block $i$ to another node not in block $i$.  Then $c_0(\omega) = \frac{1}{2}\sum_{i=1}^m |G_{b_i}(\omega)|$, and $c_2(\omega) + 2c_1(\omega) = \sum_{i=1}^m |G_{b_i}(\omega)|^2$.  By Cauchy-Schwarz inequality, we have that 
$$c_2(\omega) + 2c_1(\omega)\geq \tfrac{\left(\sum_{i=1}^m |G_{b_i}(\omega)|\right)^2}{m} = \tfrac{4c_0^2(\omega)}{m}.$$
Then, $$2c_1+c_2 = \tfrac{1}{L}\sum_{\omega\in\Omega} 2c_1(\omega) + c_2(\omega) \geq \tfrac{1}{L}\sum_{\omega\in\Omega} \tfrac{4c_0^2(\omega)}{m} \geq \tfrac{4c_0^2}{m}. $$
So $2c_1+c_2-\tfrac{4c_0^2}{m}\geq 0$.  

Also, for any node $j$ in block $i$, it is clear that $|G_{b_i}(\omega)| \leq |A_{j,L,1}|$, so $2c_1(\omega)+c_2(\omega) \leq \tfrac{1}{L^2} \sum_{j=1}^n |A_{j,n,1}|^2$.  Then $2c_1+c_2\leq \tfrac{1}{L^2} \sum_{j=1}^n |A_{j,n,1}|^2 \leq \tfrac{1}{L^2} \sum_{j=1}^n |A_{j,n,1}||A_{j,n,2}| = o(n^{1.5})$.

Notice that $a(m-a)=O(m^2)$ for $a=[mw_k], k=1,\dots,K$, and $0<a(m-a)\leq \frac{m^2}{4}$, we have 
$$-4a^2(m-a)^2+m(m-1)a(m-a)\in \left[-\tfrac{m^3}{4}, \tfrac{m^2(m-1)^2}{16}\right].$$
If $a$ is in a range such that $-4a^2(m-a)^2+m(m-1)a(m-a)=O(m^3)$, notice that $c_0 = O(|G|)$, $2c_1+c_2 = o(n^{1.5}), \tfrac{4 c_0^2}{m} = o(m c_0)$, and $c_1\geq c_0$, so term $(2c_1+c_2)\tfrac{-4a^2(m-a)^2+m(m-1)a(m-a)}{m^2(m-1)^2}$ is dominated by term $c_1\tfrac{4a^2(m-a)^2}{m^2(m-1)^2}$, and term $\left(2c_1+c_2-\tfrac{4c_0^2}{m} + \tfrac{2c_0^2}{m(m-1)}\right)\tfrac{-4a^2(m-a)^2+m(m-1)a(m-a)}{m(m-1)(m-2)(m-3)}$ is dominated by term $c_1\tfrac{a(m-a)(a-1)(m-a-1)}{m(m-1)(m-2)(m-3)}$, then 
$$\lim_{m\rightarrow\infty} \tfrac{\bV_{\CBB}(R_G(t))}{\bV_{\CBP}(R_G(t))} = 1.$$

If $a$ is in a range such that $-4a^2(m-a)^2+m(m-1)a(m-a)$ is of order higher than $m^3$, then $-4a^2(m-a)^2+m(m-1)a(m-a)$ must be positive.  Under Condition 2, we have $2c_1+c_2-\frac{4c_0^2}{m} = O(2c_1+c_2)$, then $\lim_{m\rightarrow\infty} \tfrac{\bV_{\CBB}(R_G(t))}{\bV_{\CBP}(R_G(t))}$ is a positive constant.  Under Condition 3, we have $2c_1+c_2=O(c_1)$, then $\lim_{m\rightarrow\infty} \tfrac{\bV_{\CBB}(R_G(t))}{\bV_{\CBP}(R_G(t))}$ is also a positive constant.

For $t=aL, a=[mw_k], k=1,\dots,K$, notice that $ \bE_{\CBB}(R_G(t))-E_{\CBP}(R_G(t)) = -c_0\tfrac{2a(m-a)}{m^2(m-1)}$.  From the above arguments, we know that $\bV_\CBB(R_G(t))$ is at least of order $|G|$, and $c_0=O(|G|)$, then $\frac{\bE_{\CBB}(R_G(t))-E_{\CBP}(R_G(t))}{\sqrt{\bV_\CBB(R_G(t))}}$ is at most of order $|G|^{1/2}m^{-1} = m^{0.5\alpha-1}$, which converges to 0 as $m\rightarrow\infty$ since $\alpha<\tfrac{8}{7}$.

\subsection{Proof for Theorem \ref{thm:Cov_cbp}}
\label{sec:Cov_cbp_Proof}
We need to figure out $\bE_{\CBP}(R_G(t_1)R_G(t_2))$ for $t_1=a_1L<t_2=a_2L$ with $a_1, a_2\in\{1,\dots,m\}$.  We have
\begin{align*}
& \bE_{\CBP}(R_G(t_1)R_G(t_2)) \\
& = \sum_{(i,j), (u,v)\in G} P(g_{\pi_{\CBP}(i)}(t_1)\neq g_{\pi_{\CBP}(j)}(t_1),\ g_{\pi_{\CBP}(u)}(t_2) \neq \pi_{\pi_{\CBP}(v)}(t_2)) \\
& = \sum_{(i,j)\in G} P(g_{\pi_{\CBP}(i)}(t_1)\neq g_{\pi_{\CBP}(j)}(t_1), g_{\pi_{\CBP}(i)}(t_2)\neq g_{\pi_{\CBP}(j)}(t_2)) \\
&  + \sum_{(i,j), (i,u)\in G;\ j\neq u} P(g_{\pi_{\CBP}(i)}(t_1)\neq g_{\pi_{\CBP}(j)}(t_1),\ g_{\pi_{\CBP}(i)}(t_2) \neq \pi_{\pi_{\CBP}(u)}(t_2)) \\
&  + \sum_{\tiny \begin{array}{c}(i,j), (u,v)\in G\\ i,j,u,v\text{ all different}\end{array}} P(g_{\pi_{\CBP}(i)}(t_1)\neq g_{\pi_{\CBP}(j)}(t_1),\ g_{\pi_{\CBP}(u)}(t_2) \neq \pi_{\pi_{\CBP}(v)}(t_2)).
\end{align*}

Notice that, when $i$, $j$, $u$, $v$ are all in different blocks, we have
\begin{align*}
& P(g_{\pi_{\CBP}(i)}(t_1)\neq g_{\pi_{\CBP}(j)}(t_1), g_{\pi_{\CBP}(i)}(t_2)\neq g_{\pi_{\CBP}(j)}(t_2)) \\
& \quad  = \tfrac{2a_1(m-a_2)}{m(m-1)} = q_1(a_1,a_2), \\
& P(g_{\pi_{\CBP}(i)}(t_1)\neq g_{\pi_{\CBP}(j)}(t_1),\ g_{\pi_{\CBP}(i)}(t_2) \neq \pi_{\pi_{\CBP}(u)}(t_2)) \\
& \quad  = \tfrac{a_1(m-a_2)(m-2a_1+2a_2-2)}{m(m-1)(m-2)} = q_2(a_1,a_2), \\
& P(g_{\pi_{\CBP}(i)}(t_1)\neq g_{\pi_{\CBP}(j)}(t_1),\ g_{\pi_{\CBP}(u)}(t_2) \neq \pi_{\pi_{\CBP}(v)}(t_2)) \\
& \quad = \tfrac{4a_1(m-a_2)[(a_1-1)(m-a_1-1)+(a_2-a_1)(m-a_1-2)]}{m(m-1)(m-2)(m-3)} = q_3(a_1,a_2).
\end{align*}

When some or all $i$, $j$, $u$, $v$ are in the same block, we could follow exactly the same procedure in the proof for Theorem \ref{thm:Vcbp} while replacing $p_1(a)$, $p_2(a)$, $p_3(a)$ by $q_1(a_1,a_2)$, $q_2(a_1,a_2)$, $q_3(a_1,a_2)$, respectively.  Hence, the result follows.

\end{document}